\documentclass[12pt, a4paper]{article}
\usepackage{amsmath}
\usepackage[english]{babel}
\usepackage{amsfonts}
\usepackage{enumerate}
\usepackage[pagewise,running]{lineno}
\usepackage{placeins}
\usepackage{longtable}
\usepackage{float}
\usepackage{comment}
\usepackage{color}
\usepackage{lscape}
\usepackage{graphicx}
\usepackage[utf8]{inputenc}
\usepackage{amsthm,amssymb}
\usepackage{mathrsfs}
\usepackage[top=1in, bottom=1in, left=1in, right=1in]{geometry}
\numberwithin{figure}{section}
\newtheorem{theorem}{Theorem}[section]
\newtheorem{lemma}[theorem]{Lemma}

\newtheorem{proposition}[theorem]{Proposition}
\newtheorem{example}[theorem]{Example}

\newtheorem{definition}[theorem]{Definition}

\newcommand{\RomanNumeralCaps}[1]{\MakeUppercase{\romannumeral #1}}
\DeclareMathOperator{\Aut}{Aut}
\DeclareMathOperator{\ord}{ord}
\DeclareMathOperator{\lcm}{lcm}
\DeclareMathOperator{\aut}{Aut}

\begin{document}

\title{   On $(\theta, \Theta)$-cyclic codes and their applications in constructing QECCs
\footnotetext{Email:  shuklaawadhesh@bhu.ac.in 
 (A. K. Shukla ), sachiniitk93@gmail.com (Sachin Pathak), opandey1302@bhu.ac.in 
 (O. P. Pandey), mishravipul10@bhu.ac.in (V. Mishra), upadhyay@bhu.ac.in (A. K.  Upadhyay).}}
	\author{  {Awadhesh Kumar Shukla}$^{1}$, {Sachin Pathak}$^{2}$, {Om Prakash Pandey}$^{1}$,\\ {Vipul Mishra }$^{1}$,  {Ashish Kumar Upadhyay}$^{1}$}
\date{
		\small{
	${}^{1}${Department of Mathematics, Banaras Hindu University, Varanasi, Uttar Pradesh 221005, \\India.}\\ ${}^{2}${Department of Mathematics and Basic Sciences, NIIT University, Neemrana 301705, India.}}\\
\today}
\maketitle

\begin{abstract}
Let $\mathbb F_q$  be a finite field, where $q$ is an odd prime power. Let 
$R=\mathbb{F}_q+u\mathbb{F}_q+v\mathbb{F}_q+uv\mathbb F_q$ with $u^2=u,v^2=v,uv=vu$. In this paper, we study the algebraic structure of $(\theta, \Theta)$-cyclic codes of block length $(r,s )$ over $\mathbb{F}_qR.$ Specifically, we analyze the structure of these codes as left $R[x:\Theta]$-submodules of $\mathfrak{R}_{r,s} = \frac{\mathbb{F}_q[x:\theta]}{\langle x^r-1\rangle} \times \frac{R[x:\Theta]}{\langle x^s-1\rangle}$. Our investigation involves determining generator polynomials and minimal generating sets for this family of codes. Further, we discuss the algebraic structure of separable codes. A relationship between the generator polynomials of   $(\theta, \Theta)$-cyclic codes over $\mathbb F_qR$ and their duals is established.  Moreover, we calculate the generator polynomials of the dual of $(\theta, \Theta)$-cyclic codes.  As an application of our study, we provide a construction of quantum error-correcting codes (QECCs) from    $(\theta, \Theta)$-cyclic codes of block length $(r,s)$ over $\mathbb{F}_qR$. We support our theoretical results with illustrative examples. 

\end{abstract}

\noindent {\bf{Keywords:}} \small{ $(\theta,\Theta )$-cyclic code, Generator polynomials, Dual codes, Separable Codes,   QECCs.}\par

\section{Introduction} 
Cyclic codes have been extensively researched and analyzed by researchers due to their abundant algebraic structure. This inherent mathematical richness makes cyclic codes one of the most significant classes of error-correcting codes. 
In 1994, Hammons et al. \cite{HKC94} presented a novel viewpoint by constructing non-linear codes over $\mathbb Z_2$ using the Gray map derived from codes over $\mathbb Z_4$.  Over the last three decades, scholars have dedicated significant attention to the exploration the characteristics of cyclic codes, particularly their generator polynomials.\par
In 1997, Rifà and Pojul \cite{RP97} pioneered introducing codes over mixed alphabets. Subsequently, Brouwer et al. \cite{BHOS98} examined the codes over the mixed alphabets $\mathbb Z_2$ and $\mathbb Z_3$, focusing on deriving the maximum size of error-correcting codes. After that, many scholars concentrated extensively on the mixed alphabet. In 2010, Borges et al. \cite{BFRV10} delved into the realm of mixed alphabets and examined the $\mathbb Z_2\mathbb Z_4$-additive codes. In their work, they focused not only on finding the generator matrices but also the parity check matrices and other parameters for these codes.
Building upon prior research,  Aydogdu and Siap \cite{AS13} delved deeper into the realm of  $\mathbb Z_2\mathbb Z_{2^s}$-additive codes.  Abualrub et al. \cite{ASA14}  examined  $\mathbb Z_2\mathbb Z_4$-additive cyclic codes,  with a particular focus on identifying their generator polynomials and minimum spanning sets in 2014. This research led to the construction of several Maximum Distance Separable (MDS) codes.\par
The double cyclic codes over $\mathbb Z_2$ were analyzed and examined in 2018 by Borges et al.  \cite{BCT18} for the first time. Their study involved considering $\mathbb Z_2$-double cyclic codes as $\mathbb Z_2[x]$-submodules of $\frac{\mathbb Z_2[x]}{\langle x^r-1\rangle}\times \frac{\mathbb Z_2[x]}{\langle x^s-1\rangle}$. Within this research, they successfully determined the generating sets and derived numerous optimal codes. The structural characteristics of double cyclic codes over $\mathbb Z_4$ were thoroughly studied by
Gao et al. \cite{GSWF16}. In their research, they focused on generator polynomials and minimum spanning sets and derived corresponding dual codes. Furthermore, they find several  non-linear optimal codes over $\mathbb Z_2.$  In 2020, Dinh et al.  \cite{DPBUC20}  examined the structural characteristics of cyclic codes over $\mathbb F_qR$ and their corresponding duals, where $R=\mathbb{F}_q+u\mathbb{F}_q+v\mathbb{F}_q+uv\mathbb{F}_q$, with $u^2=u, v^2=v$ and $uv=vu$. Additionally, they applied their findings to construct several novel QECCs. \par
 
In 2007, D. Boucher et al. \cite{BGU07} introduced skew cyclic codes. They derived improved linear codes with better parameters in \cite{BU009}.  Shortly thereafter, 
in 2012, Siap et al.  \cite{SAAS11} extended the study of \cite{BU009} and studied the skew cyclic codes of arbitrary lengths.  Subsequent to their groundbreaking work, several researchers delved into the exploration of skew rings of the form $R=\mathbb{F}_q+v\mathbb{F}_q$, where $v^2=v$, employing various automorphisms on $R$ (see  \cite{GJ13, GSY14, OOI19}).
 {In 2015, Yao et al. \cite{YSS15} conducted a comprehensive study on} skew cyclic codes over the ring $R=\mathbb{F}_q+u\mathbb{F}_q+v\mathbb{F}_q+uv\mathbb{F}_q$, with $u^2=u, v^2=v$ and $uv=vu$.\par

A pivotal moment in the development of quantum error correction occurred in 1995 when Shor \cite{S95} introduced the first QECC. Following this, Calderbank et al. \cite{CRSS98} introduced a technique for constructing QECCs in 1998. They achieved this by leveraging classical error-correcting codes and successfully established their existence through proof. Building upon the concept introduced by Calderbank et al. \cite{CRSS98}, numerous  QECCs have been obtained by adapting cyclic codes and their extensions, working within the realms of both finite rings and finite fields.\par
Recently, Aydogdu et al. studied double skew cyclic code over $\mathbb F_q$ in \cite{AHS22}.  Within this research paper, they examined the generator polynomials of this family of codes and their corresponding dual codes.
Skew polynomial rings allow for numerous factorizations of every polynomial, they do not have unique factorization features. This makes the study of codes over non-commutative rings especially important.\par
After being motivated from the study of \cite{AHS22}, in our paper, we consider   
$R=\mathbb{F}_q+u\mathbb{F}_q+v\mathbb{F}q+uv\mathbb{F}_q$, with $q=p^m$, where $p$ is an odd prime and $u^2=u, v^2=v$ and $uv=vu$, $m$ is a  natural number and  study structural  properties of $(\theta, \Theta)$-cyclic code over $\mathbb F_qR$ and their application in constructing in QECCs. This paper is outlined as follows:  Section 2 provides fundamental concepts, notations, and prior findings relevant to the study of  $(\theta, \Theta)$-cyclic codes of block length $(r,s )$ over $\mathbb F_qR$. Moreover,  a naturally extended Gray map from $\mathbb F_q^r\times R^s\to \mathbb F_q^{r+4s}$ is defined, and some properties of this map are established. In Section 3, we explore the generator polynomials by Theorem \ref{theo-4.4} and the minimal generating set by Theorem \ref{theo-4.12} of this family of codes.  Moreover, we discuss the generator polynomials of separable code in Theorem  \ref{theo-4.9}. Additionally, we find some near-optimal and optimal codes over $\mathbb F_9$ in Table \ref{taboptimal-1}.  In Section $4$, we conduct the study on the dual of $(\theta, \Theta)$-cyclic codes of block length $(r,s )$ over $\mathbb F_qR$ and determine the generator polynomials of dual codes in Theorem \ref {theo-5.11}.   Section $5$ focuses on constructing QECCs from separable $(\theta, \Theta)$-cyclic codes over $\mathbb F_qR$.  In Examples \ref{ex-6.8} and \ref{ex-6.9}, we provide an explanation  of constructing    QECCs from separable $(\theta, \Theta)$-cyclic codes of block length $(r,s )$ over $\mathbb F_qR$.    In Table \ref{tab-2}, we obtain several maximum distance separable (MDS) QECCs from $\theta$-cyclic code over $\mathbb F_q$ and new QECCs having better parameters than existing QECCs from $\Theta $-cyclic code over $R$ and $(\theta, \Theta)$-cyclic code over $\mathbb F_qR$ are given in Tables \ref{tab-3} and \ref{tab-4}.  In Section $6$, this paper is concluded. 

\section{Preliminaries} 

Let $\mathbb{F}_q$ be a finite field, where $q$ is a power of an odd prime $p$. The set  $\mathbb{F}_q^r$ consists of tuples of length $r$, where elements can be added and multiplied by scalars following standard rules and forms a vector space over $\mathbb F_q$. A non-empty subspace $\mathfrak{C}_r$ of $\mathbb{F}_q^r$ is called a linear code of length $r$ over $\mathbb{F}_q$, with its dimension indicating the size of $\mathfrak{C}_r$. Elements in $\mathfrak{C}_r$ are termed codewords over $\mathbb{F}_q$. The Hamming weight of any codeword $\textbf{c} = (c_0, c_1, \ldots, c_{r-1}) \in \mathfrak{C}_r$, denoted as $w_H(\textbf{c})$, counts the number of non-zero components. Furthermore, the Hamming distance between any two codewords, $\textbf{c}$ and $\textbf{c}'$, is defined as the number of positions where they differ, given by $d_H(\textbf{c},\textbf{c}') = w_H(\textbf{c} - \textbf{c}')$. The Hamming distance of a linear code $\mathfrak{C}_r$ is determined by \[d_H((\mathfrak C_r)=\mbox{min}~\{d_H(\textbf{c},\textbf{c}'),  \forall ~ \textbf{c}, \textbf{c}' \in \mathfrak C_r, ~\textbf{c}\ne \textbf{c}'\}.\] 

\begin{definition}\label{def-2.1}
Let $\aut(\mathbb F_q)$ denote the set of all automorphisms on  $\mathbb F_q$ and $\theta \in \aut(\mathbb F_q)$. Then we define the set    \[\mathbb{F}_q[x:\theta] = \{ g_0+g_1x+g_2x^2+\cdots+g_{n}x^n, g_i\in \mathbb{F}_q,~ \theta \in \aut(\mathbb{F}_q)\}.\] If $\theta$ is not an identity automorphism, then under the conventional addition and multiplication of the polynomials given by  $(gx^i)(g'x^j)= g\theta^i(g')x^{i+j}$,   the set $\mathbb{F}_q[x:\theta]$ forms a non-commutative ring and is known as skew polynomial ring over $\mathbb F_q$ (see $\cite{D74}~)$.
\end{definition}
As the set of all automorphisms on $\mathbb F_q$ constitutes a group under the composition of maps,  let $\theta\in \aut(\mathbb F_q)$ be an automorphism.   The smallest positive integer $t$ satisfying  $\theta^t(g)=g$ for all $g\in \mathbb F_q$, is called  the order of $\theta$ and  is denoted by $\ord(\theta).$ 
\begin{proposition}\label{prop-2.2}\cite{JLU12}
In  $\mathbb F_q[x:\theta]$, the following   are equivalent :

   \begin{enumerate}
       \item $x^r-1\in Z(\mathbb{F}_q[x:\theta])$, where $Z(\mathbb{F}_q[x:\theta])$ denotes the center of the ring $\mathbb F_q[x:\theta ].$
       \item  $\langle x^r-1\rangle$ forms a two-sided ideal.
       \item $\ord(\theta)$ divides $r$.
      
   \end{enumerate}
   \end{proposition}
\indent We now extend our discussion to a finite ring. We begin with the following definition.
\begin{definition}
 Let $R$ be a finite commutative ring and let the set of all automorphisms on $R$ be denoted by $\aut(R)$. Assume   $\Theta\in \aut(R)$, then \[ R[x:\Theta]=\{b_0+b_1x+b_2x^2+\cdots+b_nx^n~|~b_i\in R,~\Theta\in \Aut(R)\}.\] If $\Theta$ is not an identity automorphism then under the conventional addition and multiplication of the polynomials given by  $(bx^i)(b'x^j)= b\Theta^i(b')x^{i+j}$,   the set $R[x:\Theta]$ forms a non-commutative ring over $R$. 
  \end{definition}
  
  \begin{definition}\label{def-2.11} A   non-empty $R$-submodule $\mathfrak C_s$   of $R^s$ is called    a linear code of length $s$ over  $R$.

\end{definition}
In this context, if  ${\ord}(\Theta)$ divides $s$, the quotient space $\frac{ R[x:\Theta]}{\langle x^{s}-1\rangle}$ forms a ring. On the other hand   if $\ord(\Theta)$ does not divide $s$, then $\frac{ R[x:\Theta]}{\langle x^{s}-1\rangle}$ does not form a ring, but  it still behaves as a left $R[x:\Theta]$-module in this case. The scalar  multiplication  in this context   can be given as follows: suppose    $s(x), t(x)\in R[x:\Theta]$,\[s(x)(t(x)+\langle x^s-1\rangle)=s(x)t(x)+\langle x^s-1\rangle.\]
  \begin{definition} \label{def-2.12}
Suppose $\Theta \in \aut(R)$ and  $\mathfrak{C}_s$ is a linear code of length $s$.  If   $\textit{\textbf{b}} = (b_{0}, b_{1}, \ldots, b_{s-1}) \in \mathfrak{C}_s$, we get     $(\Theta(b_{s-1}), \Theta(b_{0}), \ldots, \Theta(b_{s-2})) \in \mathfrak{C}_s$.  Then it is called  a  $\Theta$-cyclic code over $R$.
\end{definition}
It is possible to identify 
any codeword  $\textit{\textbf{b}}=(b_0,b_1,\ldots,b_{s-1})\in \mathfrak C_s$  as a polynomial \linebreak  $(b_0+b_1x+b_2x^2+\cdots+b_{s-1}x^{s-1})\in  \frac{R[x:\Theta]}{\langle x^s-1\rangle}$. Consequently,  any $\Theta$-cyclic  code of length $s$  considered   as a left $R[x:\Theta]$-submodule of $\frac{R[x:\Theta]}{\langle x^s-1\rangle}.$\par
\vspace{.2cm}
\indent  Throughout  this work , we examine the ring $R = \mathbb{F}_q + u\mathbb{F}_q + v\mathbb{F}_q + uv\mathbb{F}_q$ with  $u^2 = u$, $v^2 = v$, and $uv = vu$, where  $q = p^m$ for an odd prime $p$ and $m$ being a positive integer.  Recall that,   $\theta\in \aut(\mathbb F_{q})$ is  given by $\theta(g)=g^{p^i}$, where $1\le i\le m$. Let us   define a map    $\Theta$ on   $R$  in the following way: 
\begin{align*}
    \Theta (a_0+ua_1+va_2+uva_3)& = \theta(a_0) + u \theta(a
_1)+v\theta(a_2)+uv\theta(a_3)\\
    &=a_0^{p^i}+ua_1^{p^i}+va_2^{p^i}+uv a_3^{p^i},
\end{align*}
where $a_0,a_1,a_2,a_3\in \mathbb F_{q}$.  By the definition of $\theta$,  it can be seen that $\Theta\in \aut(R) $.\par
\vspace{.2cm}
 Suppose 
\begin{align*}
&\kappa_1 = 1-u-v+uv, ~
    \kappa_2= u-uv,\\
    &\kappa_3=v-uv, ~
    \kappa_4=uv, 
  \end{align*}
  then 
 
 \begin{enumerate}
     \item $\kappa_1 +\kappa_2+\kappa_3+\kappa_4 = 1.$
     \item $\kappa_i^2=\kappa_i $.
     \item $\kappa_{i}\cdot \kappa_{j} = 0$ for $i\neq j.$
\end{enumerate}
      
 It is clear that   $\{\kappa_1, \kappa_2, \kappa_3 , \kappa_4\} $ is a set of   idempotents of  $R$.
By using set of  idempotents, we can express $R=\kappa_{1}R\oplus \kappa_2R\oplus\kappa_3R\oplus\kappa_4R.$ If  $b\in R$, then    \[b = a_0+ua_1+va_2+uva_3 = \kappa_1n_1+\kappa_2n_2+\kappa_3n_3+\kappa_4n_4,\] where
$n_1=a_0,~ n_2=a_0+a_1,~ n_3=a_0+a_2$ and $n_4=a_0+a_1+a_2+a_3$.\par
\vspace{.3cm}
Suppose  $\mathfrak C_s$ is a  linear code of length $s$ over  $R$ and define  the sets 
\begin{align*}
    \mathfrak C_{s,1}&=\{\textit{{\textbf n}}_1\in \mathbb F_{q}^s, \kappa_1\textit{{\textbf n}}_1+\kappa_2\textit{{\textbf n}}_2+\kappa_3\textit{{\textbf n}}_3+\kappa_4\textit{{\textbf n}}_4\in \mathfrak C_s:~\textit{\textbf n}_2,\textit{\textbf n}_3,\textit{\textbf n}_4\in \mathbb F_{q}^s\},\\
   \mathfrak C_{s,2}&=\{\textit{{\textbf n}}_2\in \mathbb F_{q}^s, \kappa_1\textit{{\textbf n}}_1+\kappa_2\textit{{\textbf n}}_2+\kappa_3\textit{{\textbf n}}_3+\kappa_4\textit{{\textbf n}}_4\in \mathfrak C_s:~\textit{\textbf n}_1,\textit{\textbf n}_3,\textit{\textbf n}_4\in \mathbb F_{q}^s\},\\
  \mathfrak C_{s,3}&=\{\textit{{\textbf n}}_3\in \mathbb F_{q}^s, \kappa_1\textit{{\textbf n}}_1+\kappa_2\textit{{\textbf n}}_2+\kappa_3\textit{{\textbf n}}_3+\kappa_4\textit{{\textbf k}}_4\in \mathfrak C_s:~\textit{\textbf n}_1,\textit{\textbf n}_2,\textit{\textbf n}_4\in \mathbb F_{q}^s\},\\
    \mathfrak C_{s,4}&=\{\textit{{\textbf n}}_4\in \mathbb F_{q}^s, \kappa_1\textit{{\textbf n}}_1+\kappa_2\textit{{\textbf n}}_2+\kappa_3\textit{{\textbf n}}_3+\kappa_4\textit{{\textbf n}}_4\in \mathfrak C_s:~\textit{\textbf n}_1,\textit{\textbf n}_2,\textit{\textbf n}_3\in \mathbb F_{q}^s\}.
\end{align*}
By the definition of  $ \mathfrak C_{s,i}$, for $i=1,2,3,4$, it is clear that $ \mathfrak C_{s,i}$  are linear codes of length $s$ over $\mathbb F_{q}.$  Moreover,  \[\mathfrak C_s=\kappa_1\mathfrak C_{s,1}\oplus\kappa_2 \mathfrak C_{s,2}\oplus\kappa_3\mathfrak C_{s,3}\oplus\kappa_4 \mathfrak C_{s,4}.\]\par
Now, let us introduce the Gray map $\phi_1: R\to  \mathbb F_q^{4}$  given by  \[\phi_1(b)= \phi_1(\kappa_1n_1+\kappa_2n_2+\kappa_3n_3+\kappa_4n_4)=(n_1,n_2,n_3,n_4)M,\] where $M\in$ GL(4, $\mathbb{F}_q$) satisfying  $MM^T=\lambda I_4$,     $\lambda (\neq 0)\in \mathbb{F}_q$,  $I_4$ is identity matrix of order $4$ and   $M^T$ denotes the transpose of matrix $M$.   For convenience, we write  $(n_1,n_2,n_3,n_4)M=bM.$ \par
It is straightforward to confirm that $\phi_1$ is $\mathbb F_q$-linear.  A natural extension of this Gray map is   $\Phi_1$ from $R^s$ to $\mathbb F_q^{4s}$, which is  given by 
\begin{multline*}  
\Phi_1 (b_{0},b_{1},\ldots,b_{s-1}) = ((n_{1,0},n_{2,0},n_{3,0},n_{4,0})M,(n_{1,1},n_{2,1},n_{3,1},n_{4,1})M,\\\ldots (n_{1,s-1},n_{2,s-1},n_{3,s-1},n_{4,s-1})M)=(b_0M,b_1M,\ldots ,b_{s-1}M),
\end{multline*}
where $b_{i} = \kappa_1n_{1,i}+\kappa_2n_{2,i}+\kappa_3n_{3,i}+\kappa_4n_{4,i}\in R$, for $i = 0,1,2\ldots,s-1$.\vskip 3pt 
Let $w_L$ and  $w_H$  denote the Lee and  Hamming weights, respectively. Then  the Lee weight of any element 
${b=\kappa_1n_1+\kappa_2n_2+\kappa_3n_3+\kappa_4n_4}\in R $  is defined as 
\begin{align*}
w_L(b)&=w_L(\kappa_1n_1+\kappa_2n_2+\kappa_3n_3+\kappa_4n_4)\\
&=w_H(\phi_1(\kappa_1n_1+\kappa_2n_2+\kappa_3n_3+\kappa_4n_4))\\
&=w_H((n_1,n_2,n_3,n_4)M).
\end{align*} 
Furthermore, if   $\textit{\textbf {b}}$ and $ \textit{\textbf {b}}^\prime$  are any two distinct elements in $R^s$, then Lee distance  is defined as  \[ d_{L}(\textit{\textbf  b} , \textit{\textbf b}^\prime ) = w_{L}( \textit{\textbf  b} - \textit{\textbf b}^\prime) = w_{H}( \Phi_1 (\textit{\textbf  b} - \textit{\textbf b}^\prime)) = d_{H}(\Phi_1 (\textit{{\textbf b}}), \Phi_1 (\textit{{\textbf  b}}^\prime)).\] 
The Lee weight of any linear code  $\mathfrak C_s$  over $R$ is the minimum  Lee weight  codeword among all the possible pairs of distinct non-zero codewords in $\mathfrak C_s.$\par
\vspace{.2cm}

\indent
 We define \[\mathfrak{R}_{r,s}= \frac{\mathbb F_{q}[x:\theta]}{\langle x^{r}-1\rangle}\times \frac{ R[x:\Theta]}{\langle x^{s}-1\rangle}~~~ \mbox{and} ~~~
  \mathbb F_{q}R= \{(g,b) , g\in \mathbb{F}_{q},~b\in R\}.\]\par 
 
   \vskip .3cm
   The set  $\mathbb F_{q}R$ constitutes a ring, where addition and multiplication are to be taken component-wise. Now, a homomorphism $\eta: R\to \mathbb F_q$ is given by   \[\eta (a_0+ua_1+va_2+uva_3)=a_0,\] where $a+ub+vc+uvw\in R.$ Moreover,
     $R$ acts   on $\mathbb F_qR$ by multiplication $\star$ as  \[h\star (g,b)=(\eta(h)  g,h  b).\]
We can  extend this scalar multiplication $\star$ on  
  $\mathbb{F}_{q}^r\times R^s$ as    \[h\star \textit{\textbf d}=(\eta(h)g_{0},\eta (h) g_{1},\ldots,\eta (h) g_{r-1}, h b_{0}, h b_{1},\ldots,h b_{s-1}),\] where $\textit{\textbf d} = (g_0,g_1,\ldots,g_{r-1}, b_0,b_1,\ldots , b_{s-1})\in \mathbb{F}^r_q\times R^s$. Indeed, the multiplication  $\star $ makes $\mathbb{F}_{q}^r\times R^s$   an $R$-module.\par
  \begin{definition}    A   linear code $\mathfrak D$ of block length $(r,s)$  over $\mathbb F_qR$ is a  
 non-empty $R$-submodule   of $\mathbb{F}_{q}^r\times R^s$.
\end{definition}
  \vspace{.3cm}
Let $(g_0,g_1,\ldots,g_{r-1}, b_0,b_1,\ldots , b_{s-1})\in \mathbb F_q^r\times R^s$ and   $\rho_{(\theta, \Theta)}$ be   a map from $\mathbb F_q^r\times R^s$ to $\mathbb F_q^r\times R^s$  given by
  \[\rho_{(\theta,\Theta)}(g_0,g_1,\ldots,g_{r-1}, b_0,b_1,\ldots , b_{s-1})=(\theta (g_{r-1}),\theta (g_0),\ldots,\theta (g_{r-2}),\Theta  (b_{s-1}),\Theta (b_0),\ldots , \Theta
  (b_{s-2})).\] 
  We refer to this map as the   $(\theta, \Theta )$-cyclic shift on $\mathbb F_q^r\times R^s.$
  \begin{definition}\label{def-2.13}
      A linear code   $\mathfrak {D}$ of block length $(r,s)$   over $\mathbb F_{q}R$, which is closed  under the map $\rho_{(\theta, \Theta )}$  is said to a      $(\theta,\Theta)$-cyclic code, i.e.,   if any   $\textbf{d}=(g_{0},g_{1},...,g_{r-1}, b_{0},b_{1},...,b_{s-1})\in \mathfrak{D},$ then we have  \[\rho_{\theta,\Theta}(\textit{\textbf{d}})~=~(\theta(g_{r-1}),\theta(g_{0}),...,\theta(g_{r-2}), \Theta(b_{s-1}),\Theta(b_{0}),...,\Theta(b_{s-2}))\in \mathfrak{D}.\]  
      \end{definition}
  \vspace{.3cm}
Now, we define a Gray  map  over  $\mathbb{F}_qR$, and also give its extension from   $\mathbb F_q^r\times R^s\to \mathbb F_q^{r+4s}.$  \par
Suppose   $(g,b) = (g,\kappa_1n_1+\kappa_2n_2+\kappa_3n_3+\kappa_4n_4)\in \mathbb{F}_qR.$ Consider the   map  $\phi$ from $\mathbb{F}_qR$ to $\mathbb{F}^5_q$  given by  \[\phi(g,b)=(g, \eta(b))=(g,(n_1,n_2,n_3,n_4)M). 
\]   This map is $\mathbb{F}_q$-linear. 
Now, an extension of $\phi$  is naturally given by   $\Phi$ from   $\mathbb{F}^r_q\times R^s$ to $\mathbb F_q^{r+4s}$  as 
 \begin{align*}
(g_0,g_1,\ldots,g_{r-1},b_0,b_1,\ldots,b_{s-1})\to &(g_0,g_1,\ldots,g_{r-1},(n_{1,0},n_{2,0},n_{3,0},n_{4,0})M,
\\&(n_{1,1},n_{2,1},n_{3,1},n_{4,1})M,
\ldots, (n_{1,s-1},n_{2,s-1},n_{3,s-1},n_{4,s-1})M)\\
&=(g_0,g_1,\ldots,g_{r-1}, b_0M,b_1M,\ldots,b_{s-1}M),
 \end{align*}
    where  $b_i=\kappa_1n_{1,i}+\kappa_2n_{2,i}+\kappa_3n_{3,i}+\kappa_4n_{4,i}$, for $i=0,1,\ldots ,s-1,$ $(g_0,g_1,\ldots,g_{r-1})\in \mathbb{F}^r_q$ and \linebreak $(b_0,b_1,\ldots,b_{s-1})\in R^s.$
    \vskip 2pt
    
Let  $(\textit{\textbf {g}},\textit{\textbf b)}\in \mathbb F^r_{q}\times R^{s}$. Then   $w_{L}(\textit{\textbf {g}},\textit{\textbf b)}= w_{H}(\textit{\textbf {g}})+w_{L}\textit{(\textbf b})$.  Let $\textit{\textbf{d}}_1, \textit{\textbf{d}}_2\in \mathbb F^{r}_{q}\times R^{s}$ be any two codewords in $\mathfrak{D}$. Then  the Lee distance  is given by   \[d_{L}(\textit{\textbf d}_{1}, \textit{\textbf d}_{2})=w_{L}(\textit{\textbf d}_{1}- \textit{\textbf d}_{2})=w_{H}( \Phi(\textit{\textbf d}_{1}- \textit{\textbf d}_{2}))= d_{H}(\Phi(\textit{\textbf d}_{1}), \Phi(\textit{\textbf{d}}_{2})).\]
 For the rest of this paper, we take   $\mbox M= \frac{1}{2}
 \begin{pmatrix}
       1& 1 & 1 &1 \\
    1 & -1 & 1& -1 \\
1 & 1& -1 & -1\\
   1 & -1& -1 & 1 \\
 \end{pmatrix}
 \in {\mbox {GL}}(4,\mathbb{F}_q)$. It is clear that ${MM}^T= I_{4\times4}.$
\begin{proposition}\label{prop-3.2}\cite{DPBUC20}
Let the Gray map $\Phi$ be as above: 
\begin{enumerate}
\item Then $\Phi$ is a distance-preserving and  $\mathbb F_q$-linear.
\item    $\Phi(\mathfrak{D})$ is a    $[r+4s,k,d_H]$ linear code, where $d_{L}=d_{H}$ over  $\mathbb F_{q}$, for a linear code $\mathfrak D$ of block length $(r,s)$ over $\mathbb F_qR$ with $|\mathfrak D|= q^k.$

\end{enumerate}
\end{proposition}

    Let  $\textit{\textbf d}=(g_0,g_1,\ldots,g_{r-1}, b_0,b_1,\ldots, b_{s-1})\in \mathbb{F}^r_q\times R^s$. Every element of     $\mathbb{F}_{q}^r\times R^s$ can be identified by a pair of polynomials in  $\mathfrak{R}_{r,s}$ as follows:
    \begin{align*}
(g_{0},g_{1},\ldots,g_{r-1}, b_{0},b_{1},\ldots,b_{s-1})\to& (g_{0}+g_{1}x+\cdots+g_{r-1}x^{r-1}, b_{0}+b_{1}x+\cdots+b_{s-1}x^{s-1})\\
&= (g(x), b(x)),
     \end{align*}
    where $g(x)\in \frac{\mathbb F_{q}[x:\theta]}{\langle x^{r}-1\rangle}$ and $b(x) \in \frac{ R[x:\Theta]}{\langle x^{s}-1\rangle}.$ This map   is called  the polynomial representation  of  elements from  $\mathbb{F}_{q}^r\times R^s$ to $\mathfrak{R}_{r,s}$. This   shows that there is   a bijection  between  $\mathbb{F}_{q}^r\times R^s$ and $\mathfrak{R}_{r,s}$.\par 
    \vspace{.3cm}
Let $(g(x),b(x))\in \mathfrak{R}_{r,s}$ and  $h(x)=h_0+h_1x+\cdots+h_{n-1}x^{n-1}\in R[x:\Theta]$. The scalar multiplication on $\mathfrak{R}_{r,s}$ is given by \[h(x)\star(g(x),b(x))=(\eta (h(x))g(x), h(x)b(x)),\] where $\eta (h(x))=\eta (h_0)+\eta (h_1)x+\cdots +\eta (h_{n-1})x^{n-1}$. Moreover, when we consider the scalar multiplication $``\star"$ the ring $\mathfrak{R}_{r,s}$ forms a left $R[x:\Theta]$-module. \par 
 \begin{proposition}\label{prop-2.14}
      Assume that  $\mathfrak{D}$ is a linear code of block length $(r,s)$   over $\mathbb F_{q}R$. If  $\mathfrak D$  is represented as a polynomial such that it is a left $R[x:\Theta]$-submodule of   $\mathfrak R_{r,s}$, then $\mathfrak D$ is  $(\theta,\Theta )$-cyclic code over $\mathbb F_qR$.
      \end{proposition}
 \begin{proof}
          Consider any arbitrary element  $\textit{\textbf{d}}=(g_{0},g_{1},\ldots ,g_{r-1}, b_{0},b_{1},\ldots ,b_{s-1})\in \mathfrak{D}\subseteq\mathbb F_q^r\times R^s.$   So  $\textit{\textbf d}$ can be identified by  polynomial $(g_0+g_1x+\cdots+g_{r-1}x^{r-1},b_0+b_1x+\cdots +b_{s-1}x^{s-1})\in \mathfrak R_{r,s}.$
          Now, we have 
          \begin{align*}  
          x\star d(x)&=(xg_{0}+xg_{1}x+\cdots +xg_{r-1}x^{r-1}, xb_{0}+xb_{1}x+\cdots +xb_{s-1}x^{s-1})\\
          &=(\theta(g_{r-1})+\theta(g_{0})x+\cdots +\theta(g_{r-2})x^{r-1}, \Theta(b_{s-1})+\Theta(b_{0})x+\cdots+\Theta(b_{s-2})x^{s-1})\in \mathfrak R_{r,s}.
          \end{align*}
        Clearly,  codeword corresponding to    $x\star d(x)$ is  $(\theta(g_{r-1}),\theta(g_{0}),...,\theta(g_{r-2}), \Theta(b_{s-1}),\Theta(b_{0}),...,\linebreak\Theta(b_{s-2}))$ and  it is a $(\theta,\Theta)$-cyclic shift of codeword $\textit{\textbf{d}}\in \mathfrak D.$ Hence, by the linearity   of $\mathfrak D$, it easily follows  that the polynomial representation of $\mathfrak D$ is a $R[x:\Theta ]$-submodule of $\mathfrak R_{r,s}.$
         
          \end{proof}
          \vspace{.3cm}
           Let $\textit{\textbf d}=(g_{0},g_{1},\ldots,g_{r-1}, b_{0},b_{1},\ldots,b_{s-1})$ and ${\textit{\textbf d}^\prime}=(g'_{0},g'_{1},\ldots,g'_{r-1}, b'_{0}, b'_{1},\ldots,b'_{s-1})\in $ $\mathbb F_q^r\times R^s.$  We define an  inner product on $\mathbb F_q^r\times R^s$  as  \[\textit{\textbf d}\cdot  \textit{\textbf d}^\prime=\kappa_1\sum_{i=0}^{r-1}g_{i} g'_{i}+\sum_{j=0}^{s-1}b_{j} b'_{j}\in R.\]
          
The  dual of a linear code $\mathfrak D$   is denoted by  $\mathfrak D^\perp$, and is    defined as       \[{\mathfrak{D}^ \perp} =\{{\textit{\textbf d}}^\prime\in \mathbb{F}_{q}^r\times R^s\mid {\textit{\textbf d}}\cdot {\textit{\textbf d}}'=0~\mbox{$\forall$}~{\textit{\textbf d}}\in\mathfrak{D}\}.\]  Moreover, if $\mathfrak{D} $ is equal to  $\mathfrak{D}^\perp$, then  
  $\mathfrak{D}$ is   called self-dual and  if $\mathfrak{D}\subseteq \mathfrak{D}^\perp$, then  
  $\mathfrak{D}$ is called self-orthogonal.\par
 Now we establish the relation between $\mathfrak D$ and its dual $\mathfrak D^\perp.$

  \begin{proposition}
      Suppose  $ \ord(\theta)\mid  r,~ \ord(\Theta )\mid  s$ and $\mathfrak{D}$ is   a $(\theta, \Theta )$-cyclic  code of  block length $(r,s)$ over $\mathbb F_qR$, then $\mathbf{\mathfrak{D}^\perp}$  is also a $(\theta, \Theta)$-cyclic  code of the same  block length over $\mathbb F_{q}R$.
      \end{proposition}
      
  \begin{proof}   
   Suppose  ${\textit{\textbf d}}=(g_{0},g_{1},\ldots,g_{r-1}, b_{0},b_{1},\ldots,b_{s-1}) \in \mathfrak{D}$ and  
      since $\mathfrak{D}$ is a $(\theta, \Theta )$-cyclic code, hence     $( \theta (g_{r-1}) ,\theta (g_{0}),\ldots ,\theta (g_{r-2}), \Theta (b_{s-1}), \Theta (b_{0}), \ldots, \Theta (b_{s-2})) \in \mathfrak{D}$.
     Now, consider \linebreak${\textit{\textbf d}}'= (g'_{0},g'_{1},\ldots,g'_{r-1} , b'_{0},b'_{1},\ldots,b'_{s-1})  \in \mathfrak {D^\perp}$. 
      Then   ${\textit{\textbf d}}\cdot\textit{{\textbf d}}'= 0.$ \par 
    In order to check that   $\mathfrak{D}^\perp$ is also a   $(\theta , \Theta )$-cyclic code of block  length $(r,s)$ over $\mathbb F_qR$, it is enough to prove that 
   \[ \rho_{(\theta,\Theta)} ({\textit{\textbf d}}') = (\theta (g'_{r-1}),\ldots, \theta (g'_{r-2}), \Theta (b'_{s-1}),\ldots, \Theta (b'_{s-2})) \in \mathbf{\mathfrak{D}^\perp}.\]
     In other words, we need to have 
      \[{\textit{\textbf d}}\cdot \rho_{(\theta,\Theta)} ({\textit{\textbf d}}') = 0 ,\]
      i.e.,  \[ \kappa_1g_{0} \theta (g_{r-1}) + \kappa_1g_{1} \theta (g_{0}) +\cdots+ \kappa_1g_{r-1} \theta (g_{r-2}) +b_0 \Theta (b'_{s-1}) + b_{1} \Theta (b'_{0}) +\cdots+ b_{s-1} \Theta (b'_{s-2}) =0 .\] 
    As $\mathfrak D$ is a $(\theta,\Theta)$-cyclic code, so we have 
       \[\rho_{(\theta,\Theta)} ({\textit{\textbf d}}) \in \mathfrak{D} \mbox{~and}~\ 
      \rho_{(\theta,\Theta)} ({\textit{\textbf d}}) \cdot {\textit{\textbf d}}'= 0.\]
       Let $l= \lcm(r,s)$. Then
       \[\rho ^{l-1}_{(\theta,\Theta )} ({\textit{\textbf d}})= (\theta ^{l-1}(g_{1}),\ldots,\theta ^{l-1} (g_{0}), \Theta ^{l-1}(b_{1}),\ldots, \Theta ^{l-1}(b_{0}))\in \mathfrak D.\]  
       Now, we get 
       \begin{align*}
           0&= \rho ^{l-1}_{(\theta , \Theta )} ({\textit{\textbf d}}) \cdot {\textit{\textbf d}}'\\
           &= \kappa_1(\theta ^{l-1}(g_{1})g'_{0} + \theta^{l-1}(g_{2})g'_{1} +\cdots + \theta ^{l-1} (g_{r-1})g_{r-2} + \theta ^{l-1}(g_{0}) g'_{r-1})\\
           &\hspace{2cm}+ (\Theta ^{l-1}(b_{1})b_0 + \Theta ^{l-1}(b_{2}) b'_{1} +\cdots+ \Theta ^{l-1}(b_{s-1}) b'_{s-2} + \Theta ^{l-1}(b_{0}) b_{s-1}).
       \end{align*}
       Since $ \ord(\theta)\mid r$ and $\ord(\Theta)\mid s$, we have  $\theta ^l$ and  $\Theta ^l$ are identity automorphisms and $\theta = \Theta $ over $\mathbb F_q. $ Now, applying $\Theta$ to   both sides in the above equation
       \[\kappa_1(g_{1} \theta (g'_{0}) + g_{2} \theta (g'_{1}) +\cdots+ g_{r-2} \theta (g'_{r-1}) + g_{r-1} \theta (g'_{0}) ) + (b_{1} \Theta (b'_{0}) +\cdots+ b_{s-2} \Theta (b'_{s-1}) + b_{s-1} \Theta (b'_{0}) ) =0. \]
       Hence, ${\textit{\textbf d}} \cdot \rho _{(\theta, \Theta )} ({\textit{\textbf d}}') = 0 $. Thus, we infer that    $\rho_{(\theta,\Theta)} ({\textit{\textbf d}}) \in {\mathfrak{D}^\perp}$. This completes our proof.   
       \end{proof}

\section{ Algebraic structure of  $(\theta, \Theta)$-cyclic codes over $\mathbb F_qR$}

Let us begin by providing the decomposed structure of linear codes $\mathfrak  D$ of block length $(r,s)$ over  $\mathbb F_qR$. Let $\mathcal{D}_i$ be sets, for $i=1,2,3,4.$ We define \[\mathcal{D}_1\oplus\mathcal{D}_2\oplus\mathcal{D}_3\oplus\mathcal{D}_4= \{d_1+d_2+d_3+d_4\mid d_i\in \mathcal{D}_i~ {\mbox {for i= 1,2,3,4}}\}.\]
Since, $\sum_{i=1}^4\kappa_i=1$, so  any   $\textit{\textbf d}=(\textit{\textbf{g}},\textit{\textbf{b})}\in \mathbb{F}^r_q\times R^s$ can be expressed as   $\textit{\textbf d}=\sum_{i=1}^4\kappa_i(\textit{\textbf{g}}, \textit{\textbf{b}})\in \mathbb{F}^r_q\times R^s$, where $\textit{\textbf{g}}\in \mathbb{F}_q^r$ and $\textit{\textbf{b}}=\kappa_1\textit{{\textbf{n}}}_1+\kappa_2\textit{{\textbf{n}}}_2+\kappa_3\textit{{\textbf{n}}}_3+\kappa_4\textit{{\textbf{n}}}_4\in R^s.$ We define 

  \begin{align*}
    \mathfrak{D}_1&=\{(\textit{\textbf g},\textit{\textbf{n}}_1)\in \mathbb F_{q}^r\times \mathbb F_q^{s}\mid \textit{\textbf g}\in \mathbb F_{q}^r, ~{\textit{\textbf{n}}}_1\in \mathfrak C_{s,1}\},\\
    \mathfrak{D}_2&=\{(\textit{\textbf g},\textit{\textbf{n}}_2)\in \mathbb F_{q}^r\times \mathbb F_q^{s}\mid  \textit{\textbf g}\in \mathbb F_{q}^r, ~{\textit{\textbf{n}}}_2\in \mathfrak C_{s,2}\},\\
    \mathfrak{D}_3&=\{(\textit{\textbf g},\textit{\textbf{n}}_3)\in \mathbb F_{q}^r\times \mathbb F_q^{s}\mid  \textit{\textbf g}\in \mathbb F_{q}^r, ~{\textit{\textbf{n}}}_3\in \mathfrak C_{s,3}\},\\
    \mathfrak{D}_4&=\{(\textit{\textbf g},\textit{\textbf{n}}_4)\in \mathbb F_{q}^r\times \mathbb F_q^{s}\mid  \textit{\textbf g}\in \mathbb F_{q}^r, ~{\textit{\textbf{n}}}_4\in \mathfrak C_{s,4}\}.
\end{align*}  
It is clear  that $\mathfrak{D}_i$, for $i=1,2,3,4$ is a   linear code of block length $(r,s)$ over $\mathbb{F}_q.$ Consequently, any linear code $\mathfrak{D}$ of block length $(r,s)$ over $\mathbb{F}_qR$ has a unique expression as    \[\mathfrak{D} = \kappa_1\mathfrak{D}_1\oplus\kappa_2\mathfrak{D}_2\oplus \kappa_3 \mathfrak{D}_3\oplus \kappa_4\mathfrak{D}_4.\]

Let $c(x),~ a(x)(\neq 0)$ be two polynomials in $\mathbb{F}_q[x:\theta].$ Then  $a(x)$ is called  a right divisor of $c(x)$,   
 if  $c(x)= d(x)a(x)$, where $d(x)\in \mathbb{F}_q[x:\theta]$ and denoted by $a(x)\mid_r c(x).$ Similarly, we can define left divisor of $c(x).$ 
 \begin{definition}
     
The polynomial $c(x)$ is central in  $\mathbb F_q[x:\theta]$, if 
$xc(x)=c(x)x$.  
 \end{definition}
\begin{proposition}\label{prop-2.4}\cite[Proposition-2.3]{JLU12}
  Suppose  $c(x)$, $d(x)\in \mathbb{F}_q[x:\theta]$ are monic central  polynomials.  Then $d(x)c(x) = c(x)d(x)$.
\end{proposition}

\begin{lemma}\label{lem-2.5} \cite[Theorem-\RomanNumeralCaps{2}.11] {D74}
Suppose  $c(x)$ and $a(x)(\mbox{~monic~})\in $  
 $\mathbb F_q[x:\theta]$. Then  $c(x)=q_1(x)a(x)+q_2(x)$, where $q_2(x)=0$ or  $\deg(q_2(x))<\deg(a(x))$, for some      $q_1(x)$ and $q_2(x)$ in $\mathbb F_q[x:\theta]$.

\end{lemma}
In the context of $R[x:\Theta]$, we now introduce right division algorithm.

\begin{proposition}\label{prop-2.6} 
Let $s(x) \mbox{and}~t(x)\in R[x:\Theta]$ be two polynomials such that  $t(x)$ is monic. Then $s(x)=q'_1(x)t(x)+q'_2(x)$, where $ q'_2(x)=0$ or $ \deg(q'_2(x))<\deg(t(x)),$ for some $q'_1(x),~ q_2'(x) \in R[x:\Theta]$.  
\end{proposition}
\begin{proof}
Let $s(x) = \sum_{i=0}^{m} s_i x^i$ and $t(x) = \sum_{j=0}^{k-1} t _j x^j+x^k$. 
 Then,  $\deg(s(x) - {s_m}x^{m-k}t(x))$  is    less than   $\deg(s(x))$. By repeatedly applying this procedure with successive polynomials, we can get the desired  $q'_1(x)$ and $q'_2(x)$. 
Now, we prove the uniqueness of  $q'_1(x)$ and $q'_2(x)$.   Suppose \[s(x) = q_1^\prime(x) t(x) + q_2^\prime(x) = q_1^{\prime\prime} (x)t(x) + q_2^{\prime\prime}(x),\] then \[(q_1^\prime(x) - q_1^{\prime\prime}(x))t(x) = q_2^{\prime\prime}(x) - q_2^\prime(x).\] Here, $q_1^\prime(x) - q_1^{\prime\prime}(x) \neq 0$,  then the right polynomial must have a degree equal to or greater than the degree of  $t(x)$ and at the same time, the degree of the right polynomial will not be more than  $\deg(t(x))-1$. This leads to a contradiction.  Therefore, $q_1^\prime(x) $ and $q_2^\prime(x) $ are unique.
\end{proof}
Now, we present the definition of least common multiple $(\lcm)$ and greatest common divisor $(\gcd )$, which are used frequently throughout this paper.

\begin{definition}\label{def-2.7} Suppose polynomials   $a(x),~c(x), ~ d(x)\in \mathbb F_q[x:\theta]$ with  $d(x)\mid_r~a(x)$ and $ d(x)\mid _r~c(x)$.  If $d'(x)$ is another polynomial  satisfying  same property, then $d'(x)\mid_r~d(x)$. Consequently, $d(x)$ is said to be right gcd and is  
 denoted by  $\gcd_r(a(x), c(x))$. Similarly, we can define the left gcd $(\gcd_l)$.

\end{definition}

\begin{definition}\label{def-2.8}
Suppose polynomials $m(x),a(x), c(x)\in \mathbb F_q[x:\Theta]$ with  $a(x)\mid_l m(x)$ and $c(x)\mid_l m(x)$  and if any other  polynomial $m'(x) \in \mathbb F_q[x:\theta]$   satisfying same property implies $m(x)\mid_l~ m'(x)$. Then $m(x)$ is said to be left lcm and is  denoted by $\lcm_l(a(x),c(x))$.

\end{definition}

\begin{lemma}\label{lem-2.9}\cite[pp.486]{O33}
Let $ d(x)$ and $c(x)\in \mathbb F_q[x:\theta]$. Then  $\deg(\lcm_l(d(x),c(x))= \deg(d(x))\linebreak+\deg(c(x))-\deg(\gcd_r(d(x),c(x))).$
\end{lemma}

\begin{lemma}\label{lem-2.10}\cite[Lemma-2.7]{AHS22} Suppose   $d_1(x),d_2(x),d_3(x)$ and $d(x)\in \mathbb F_q[x:\theta]$, where $d(x)$ is  a central polynomial  with   \[d_1(x)d_2(x)\equiv0\pmod{d(x)},\] and \[d_1(x)d_3(x)\equiv 0\pmod{d(x)}.\] Then \[d_1(x){\gcd}_l(d_2(x),d_3(x))\equiv 0 \pmod{d(x)}.\]  
\end{lemma}
Following  propositions  characterize $\Theta$-cyclic codes over $R$:

\begin{proposition}\cite[Theorem-4.3]{YSS15}
Let 
 $\mathfrak C_s=\kappa_1\mathfrak C_{s,1}\oplus\kappa_2\mathfrak C_{s,2}\oplus\kappa_3\mathfrak C_{s,3}\oplus\kappa_4\mathfrak C_{s,4}$ be a linear code of  length $s$ over $R$. Then  $\mathfrak C_s$ is $\Theta$-cyclic code over $R$ if and only if  $\mathfrak C_{s,i}$ are $\theta$-cyclic codes  over $\mathbb F_q.$
\end{proposition}

\begin{proposition}\label{prop-4.3}
\cite[Theorem-4.5]{YSS15} Let 
 $\mathfrak C_s=\kappa_1\mathfrak C_{s,1}\oplus\kappa_2\mathfrak C_{s,2}\oplus\kappa_3\mathfrak C_{s,3}\oplus\kappa_4\mathfrak C_{s,4}$ and   $\mathfrak C_{s,i}=\langle t_i(x)\rangle$ with $t_i(x)\mid_r (x^s-1)$  in $R[x:\Theta]$ for  $i=1,2,3,4$.   Then $\mathfrak C_s=\langle t(x)\rangle$, where $t(x)= \kappa_1t_1(x)+\kappa_2t_2(x)+\kappa_3t_3(x)+\kappa_4t_4(x)$ with $t(x)\mid_r~ x^s-1$ in $ R[x:\Theta].$  Additionally, $|\mathfrak C_s|= q^{4s-\sum_{i=1}^4\deg(t_i(x))}.$ 
\end{proposition}
In next result, we present a decomposition   of $(\theta, \Theta )$-cyclic code over $\mathbb F_qR.$
\begin{proposition}\label{prop-3.4}
Let $\mathfrak{D}=\kappa_1\mathfrak{D}_1\oplus\kappa_2\mathfrak{D}_2\oplus \kappa_3\mathfrak{D}_3\oplus \kappa_4\mathfrak{D}_4$ be  a linear code of  block length $(r,s)$ over $\mathbb{F}_qR$. Then $\mathfrak D$ is  $(\theta,\Theta)$-cyclic code if and only if each $\mathfrak{D}_i$   is a $\theta$-cyclic code of block length $(r,s)$ over  $\mathbb{F}_q$ for $i=1,2,3,4$. 
\end{proposition}

\begin{proof}
    Suppose  $\textit{\textbf d}= (g_0,g_1,\ldots,g_{r-1} , b_0,b_1,\ldots,b_{s-1})\in \mathfrak{D}$, where $b_i=\kappa_1n_{1,i}+\kappa_2n_{2,i}+\kappa_3n_{3,i}+\kappa_4n_{4,i}$ for $i=0,1,\ldots,s-1.$ Then we have 
    \begin{align*}
       (g_0,g_1,\ldots,g_{r-1},n_{1,0},n_{1,1},\ldots,n_{1,s-1})&\in \mathfrak{D}_1\\
       (g_0,g_1,\ldots,g_{r-1},n_{2,0},n_{2,1},\ldots,n_{2,s-1})&\in \mathfrak{D}_2\\
       (g_0,g_1,\ldots,g_{r-1},n_{3,0},n_{3,1},\ldots,n_{3,s-1})&\in \mathfrak{D}_3\\
       (g_0,g_1,\ldots,g_{r-1},n_{4,0},n_{4,1},\ldots,n_{4,s-1})&\in \mathfrak{D}_4.
    \end{align*}
       If  $\mathfrak{D}$ is a    $(\theta ,\Theta)$-cyclic code, then     \[\rho_{(\theta,\Theta)}(\textit{\textbf d})=( \theta (g_{r-1}) ,\theta (g_{0}),\ldots,\theta (g_{r-1}), \Theta (b_{s-1}), \Theta (b_{0}), \Theta (b_{1}),\ldots, \Theta (b_{s-2})) \in \mathfrak{D}.\]
    Therefore,
    \begin{align*}
       (\theta (g_{r-1}),\theta (g_0),\ldots,\theta (g_{r-2}),\theta (n_{1,s-1}),\theta (n_{1,0}),\ldots,\theta (n_{1,s-2}))&\in \mathfrak{D}_1\\
       (\theta (g_{r-1}),\theta (g_0),\ldots,\theta (g_{r-2}),\theta (n_{2,s-1}),\theta(n_{2,0}),\ldots,\theta (n_{2,s-2}))&\in \mathfrak{D}_2\\
       (\theta(g_{r-1}),\theta(g_0),\ldots,\theta(g_{r-2}),\theta(n_{3,s-1}),\theta(n_{3,0}),\ldots,\theta(n_{3,s-2}))&\in \mathfrak{D}_3\\
       (\theta(g_{r-1}),\theta(g_0),\ldots,\theta(g_{r-2}),\theta(n_{4,s-1}),\theta(n_{4,0}),\ldots,\theta(n_{4,s-2}))&\in \mathfrak{D}_4.
    \end{align*} 
    Consequently, it follows that  $\mathfrak{D}_i$ for $i=1,2,3,4$ is    $\theta$-cyclic code of   block length  $(r,s)$ over $\mathbb{F}_q$.\par
 The converse part can be proved by using similar arguments.
\end{proof}

Now, we determine the generator polynomials for  $(\theta, \Theta)$-cyclic codes of block length $(r,s)$  over $\mathbb F_qR$. Additionally, we investigate the size of this family of codes and examine the minimum generating set.

\begin{theorem}\label{theo-4.4}
    Let  $\mathfrak{D}$ be a  $(\theta, \Theta)$-cyclic code of block length $(r,s)$ over $\mathbb{F}_qR$. Then \[\mathfrak{D} = \langle(\ell(x), 0), (s(x),t(x)\rangle,\] where the  polynomials   $\ell(x), s(x)\in \mathbb{F}_q[x:\theta]$ with  $\ell(x)\mid _r(x^r-1)$ and   $t(x)= \kappa_1t_1(x) + \kappa_2t_2(x)+\kappa_3t_3(x)+\kappa_4t_4(x)$,~  $t_i(x)\in \mathbb 
 F_q[x:\theta]$ with $t_i(x)\mid_r(x^s-1)$ in  $\mathbb F_q[x:\theta ],$  for  $i=1,2,3,4$ and $t(x)\mid_r (x^s-1)$ in $R[x:\Theta ]$.  
    \end{theorem}
\begin{proof}
    As both  $\mathfrak{D}$ and $\frac{ R[x:\Theta]}{\langle x^s-1\rangle}$  are left $R[x:\Theta]$-submodule of $\mathfrak{R}_{r,s}$. We define a left $R[x:\Theta]$-module homomorphism $\pi $ as follows:
    
     \[\pi :\mathfrak{D}\to \frac{ R[x:\Theta]}{\langle x^s-1\rangle}\]   given by \[\pi(d_1(x), d_2(x)) = d_2(x).\]  Clearly, $\pi(\mathfrak{D})$ is  a left-submodule of    $\frac{ R[x:\Theta]}{\langle x^s-1\rangle}$. 
     Then by Proposition \ref{prop-4.3}, we get that  $\pi(\mathfrak{D}) =\langle t(x)\rangle.$ Also, Ker$ (\pi) =\{(d_1(x), 0)\in \mathfrak{R}_{r, s}\mid  (d_1(x), d_2(x))\in \mathfrak{D}\}.$ Consider  a set    \[K = \{f(x)\in{\mathbb{F}_q[x:\theta]} \mid (f(x), 0)\in \mbox{Ker}(\pi)\}.\] 
      The set  $K$ is a left-submodule   of $\frac{\mathbb{F}_q[x:\theta]}{\langle x^r-1\rangle}$. Therefore,   $K = \langle \ell(x)\rangle$ with $\ell(x)\mid_r~(x^r-1)$. Hence, any  $(f(x), 0)\in \mbox{Ker}(\pi)$ and $f(x)  = \mu(x)\ell(x)$, where    $\mu(x)\in \mathbb F_q[x:\theta ]$. Consequently, we can express $(f(x),0) $ as  $\lambda(x)\star(\ell(x),0).$ 
    Thus,  Ker($\pi)=\langle\ell(x),0\rangle$ is a left $R[x:\Theta]$-submodule of $\mathfrak{D}$. Moreover,
    \[\frac{\mathfrak{D}}{\mbox{Ker}(\pi)}\cong \pi(\mathfrak{D})=\langle t(x)\rangle .\] Suppose $(s(x), t(x))\in \mathfrak{D}$, then $\pi(s(x),t(x)) = t(x).$ Therefore, 
 we conclude that the generator polynomials  of   $(\theta,\Theta)$-cyclic code of block length $(r,s)$ over $\mathbb{F}_qR$   are    $(\ell(x),0)$ and $ (s(x),t(x))$ . More precisely,  $\mathfrak{D} = \langle(\ell(x), 0), (s(x),t(x)\rangle$.\par
 \end{proof}
    \begin{lemma}\label{lem-3.13}
        Let $\mathfrak{D} = \langle(\ell(x), 0), (s(x),t(x)\rangle$, where $x^r-1=f(x)\ell(x)$,  $t(x)= \kappa_1t_1(x) + \kappa_2t_2(x)+\kappa_3t_3(x)+\kappa_4t_4(x)$  and  $x^s-1= h_i(x)t_i(x)$ for $i=1,2,3,4.$ Then $\deg s(x)< \deg \ell(x)$,  $\ell(x)\mid_r~h_1(x) s(x)$
    and $\lcm_l~(s(x), \ell(x))\mid_r~h_1(x) s(x)$. 
    \end{lemma}
    \begin{proof} Assume  $\deg(s(x))>\deg(\ell(x))$, then  by using division algorithm in $\mathbb{F}_q[x:\theta],$ we have  \[s(x) = \mu_1(x)\ell(x)+\mu_2(x),\] where $\mu_2(x) = 0$ or $\deg (\mu_2(x))<\deg(\ell(x)).$
    \begin{align*}
    (\mu_2(x),t(x))&=(s(x)-\mu_1(x)\ell(x),t(x))\\
    &=(s(x),t(x))-\mu_1(x)(\ell(x),0)\in \mathfrak{D}\\
    (s(x),t(x))&= (\mu_2(x),t(x))+\mu_1(x)(\ell(x),0).
    \end{align*}
   Therefore,   $\deg(s(x))<\deg(\ell(x)).$
    \par
    \noindent Now, 
    \begin{align*}
       (\kappa_1h_1(x)+\kappa_2h_2(x)+\kappa_3h_3(x)+\kappa_4h_4(x))\star(s(x),t(x))&=(h_1(x)s(x),0)\\
        &=\langle(\ell(x),0)\rangle,
    \end{align*}
    where $t(x)=\kappa_1t_1(x)+\kappa_2t_2(x)+\kappa_3t_3(x)+\kappa_4t_4(x).$
    Thus, $\ell(x)\mid_rh_1(x)s(x).$\par
Since  $\ell(x)\mid_r h_1(x)s(x)$  and $s(x)\mid_r h_1(x)s(x)$. Hence,   $\lcm_l~(\ell(x),s(x))\mid_rh_1(x)s(x).$
    \end{proof}
    
    Now,  we present the algebraic structure of separable codes. 
 \begin{definition}
     Let $\mathfrak{D} $ be  
 a linear code of block length $(r,s)$ over $\mathbb{F}_qR$. Suppose the set   $\mathfrak C_r$ is obtained by deleting the last $s $ coordinates and    $\mathfrak C_s$ is obtained by deleting the first $r$ coordinates from $\mathfrak D$.    It is worth noting that $\mathfrak C_r$ and $C_s$ are linear codes over $\mathbb F_q$ and $R$, respectively. The code   $\mathfrak{D}$ is called  a  separable code if  $\mathfrak{D}= \mathfrak  C_r\times \mathfrak C_s$.  
    \end{definition}
 
   \begin{theorem}\label{theo-4.6}
   Let $\mathfrak{D}=\mathfrak C_r\times \mathfrak C_s$ be a separable code, then $\mathfrak{D}$
 is  a       separable  $(\theta,\Theta)$-cyclic code if and only if $\mathfrak C_r$ and  $\mathfrak C_s$  are  $\theta$ and $ \Theta $-cyclic codes over $\mathbb F_q$ and $R$, respectively.  
   \end{theorem}
   \begin{proof}
    Assume that  $\mathfrak{D}$ is a separable  $(\theta, \Theta)$-cyclic  code  over $\mathbb F_qR$ and \linebreak $ (g_0, g_1,\ldots, g_{r-1}, b_0,b_1,\ldots, b_{s-1}) \in \mathfrak{D}$, then $(g_0, g_1,\ldots, g_{r-1})\in \mathfrak C_r$, $(b_0,b_1,\ldots, b_{s-1})\in \mathfrak C_s.$  Therefore,  $(\theta(g_{r-1}), \theta (g_0),\ldots, \theta (g_{r-2}) , \Theta (b_{s-1}),\Theta (b_0),\ldots, \Theta(b_{s-2}))\in \mathfrak{D},$ which implies that  \linebreak $ (\theta(g_{r-1}), \theta(g_0),\ldots, \theta(g_{r-2}))\in \mathfrak C_r$ and $(\Theta (b_{s-1}),\Theta (b_0),\ldots,\theta (b_{s-2}))\in \mathfrak C_s.$ Hence, the result follows for one side.\par  
    For converse, assume that  $\mathfrak C_r$ and  $\mathfrak C_s$  are  $\theta$ and $ \Theta $-cyclic codes over $\mathbb F_q$ and $R$, respectively and $(g_0, g_1,\ldots, g_{r-1})\in \mathfrak C_r, (b_0,b_1,\ldots, b_{s-1})\in \mathfrak C_s$, then $ (\theta(g_{r-1}), \theta(g_0),\ldots, \theta(g_{r-2}))\in \mathfrak C_r$ and  $(\Theta (b_{s-1}),\Theta (b_0),\ldots,\Theta (b_{s-2}))\in \mathfrak C_s.$
       Therefore, \[(\theta(g_{r-1}), \theta(g_0),\ldots, \theta(g_{r-2}),\Theta (b_{s-1}),\Theta (b_0),\ldots,\theta (b_{s-2}))\in  \mathfrak{D} = \mathfrak C_r\times \mathfrak C_s.\] Therefore, it  confirms that   $\mathfrak{D}$ is a separable  $(\theta,\Theta)$-cyclic code over $\mathbb F_qR$.
   \end{proof}
  Now, we explicitly present the generator polynomials of $\mathfrak C_r$ and $\mathfrak C_s$.
   \begin{theorem}\label{theo-4.7}
    Let  $\mathfrak{D}=\langle(\ell(x),0),(s(x),t(x))\rangle$ .  Then $\mathfrak C_r=\langle \gcd_r(\ell(x),s(x))\rangle$ and $\mathfrak C_s=\langle t(x)\rangle$.
   \end{theorem}
   
  \begin{proof}
      Suppose  $c(x)\in \mathfrak C_r$, then for some  $e(x)\in \frac{R[x:\Theta]}{\langle x^s-1\rangle}$, we have    $(c(x),e(x))\in \mathfrak{D}$ . Since \linebreak $\mathfrak{D}=\langle(\ell(x), 0), (s(x),t(x)\rangle\ $,  then  $(c(x),e(x))=\mu(x)\star(\ell(x),0)+\nu(x)\star(s(x),t(x))$ for some   $\mu(x),\nu(x)\in R[x:\Theta]$. This shows that $c(x)=\eta(\mu(x))\ell(x)+\eta(\nu(x))s(x)$. It follows that  $\gcd _r(\ell(x),s(x))\mid_r~ c(x).$ Consequently, $c(x)\in \langle\gcd _r(\ell(x),s(x))\rangle.$ Thus,  $\mathfrak C_r\subseteq\langle\gcd_r(\ell(x),s(x))\rangle.$\par
      \vskip 3pt
     Conversely, since  ${\gcd}_r(\ell(x),s(x))=\mu(x)\ell(x)+\nu(x)s(x)$   for some  $\mu(x),~\nu(x)\in\mathbb F_q[x:\theta],$  we have \[{\gcd}_r((\ell(x),s(x)),\nu(x)s(x))=\mu(x)\star(\ell(x),0)+\nu(x)\star(s(x),t(x))\in \mathfrak{D}.\] This shows that $\langle\gcd_r(\ell(x),s(x))\rangle\subseteq \mathfrak C_r.$ Hence, we get  $\mathfrak C_r=\langle\gcd_r(\ell(x),s(x))\rangle.$ We 
 can now  determine  $\mathfrak C_s=\langle t(x)\rangle$ by using similar steps. 
  \end{proof}
   
 \begin{proposition}\label{prop-4.8}
     Let  $\mathfrak{D} = \langle(\ell(x), 0),(s(x), t(x))\rangle$. Then  $\ell(x)~|_r~s(x)$ if and only if $s(x) = 0$.
 \end{proposition}
 \begin{proof}
     Assume that $s(x) = 0$, then it clearly holds.\par 
     Conversely, suppose    $\ell(x)~|_r~s(x)$,  then   $s(x) = \nu(x)\ell(x)$ for some polynomial  $\nu(x)\in \mathbb{F}_q[x:\theta]$. Now, let  \[\mathfrak{D}' =  \langle(\ell(x), 0),(0, t(x))\rangle.\] So,  $(0, t(x)) = (s(x),t(x))-\nu(x)\star(\ell(x),0)\in \mathfrak{D}$, implies that $\mathfrak{D}'\subseteq \mathfrak{D}$.
     Moreover,   $(s(x),t(x)) = \nu(x)\star(\ell(x),0)) + (0,t(x))\in \mathfrak{D}'$, indicating  that $\mathfrak{D}\subseteq \mathfrak{D}'.$  Therefore, we conclude that $s(x) = 0$. This proves the result. 
 \end{proof}
 We sum up our earlier findings in the following theorem:
  
 \begin{theorem}\label{theo-4.9}
      Let   $\mathfrak{D} = \langle(\ell(x), 0),(s(x), t(x))\rangle$. The  following are  equivalent statements: 
      \begin{enumerate}
        \item $\mathfrak C_r = \langle\ell(x)\rangle,\mathfrak C_s = \langle t(x)\rangle$;\
        \item $\ell(x)~|_r~s(x)$;\
         \item $\mathfrak{D}$ is a separable code;\

           \item $\mathfrak{D} = \langle(\ell(x),0),(0,t(x))\rangle$.

          \end{enumerate}
 \end{theorem}
\begin{proof}
    Theorem \ref{theo-4.7} and Proposition \ref{prop-4.8}  together give the proof.
\end{proof}
In the following theorem, we determine the minimal generating set for $(\theta,\Theta)$-cyclic codes over $\mathbb{F}_qR$. This investigation is crucial for understanding the fundamental properties and structures of these codes. By characterizing the minimal generating sets, we can gain deeper insights into the algebraic and combinatorial aspects of $(\theta,\Theta)$-cyclic codes. 

\begin{theorem}\label{theo-4.12}
Let   $\mathfrak{D}= \langle(\ell(x),0),(s(x),t(x))\rangle$, where $x^r-1=f(x)\ell(x),~t(x)=\kappa_1t_1(x) + \kappa_2t_2(x)+\kappa_3t_3(x)+\kappa_4t_4(x)~\mbox{and}~ x^s-1= h_i(x)t_i(x),~\mbox{for}~i= 1,2,3,4. $ Suppose
\begin{align*}
    &\nonumber G_1= \bigcup\limits_{i=0}^{\deg(f(x))-1}\{x^i\star(\ell(x),0)\},\\
&\nonumber G_2= \bigcup\limits_{i=0}^{\deg(h_1(x))-1}\{x^i\star(s(x),\kappa_1t_1(x))\},\\
&\nonumber G_3= \bigcup\limits_{i=0}^{\deg(h_2(x))-1}\{x^i\star(0,\kappa_2t_2(x))\},\\
&\nonumber G_4= \bigcup\limits_{i=0}^{\deg(h_3(x))-1}\{x^i\star(0,\kappa_3t_3(x))\},\\
&\nonumber G_5= \bigcup\limits_{i=0}^{\deg(h_4(x))-1}\{x^i\star(0,\kappa_4t_4(x))\}.
\end{align*}
Then $G= G_1\cup G_2\cup G_3\cup G_4\cup G_5$ forms    a minimal generating set of   $\mathfrak{D}$. Additionally, $|\mathfrak{D}|= q^{\deg(f(x))+\sum_{i=1}^4\deg(h_i(x))}$. 
\end{theorem}

\begin{proof}
Suppose  $d(x)\in \mathfrak{D}$,  then 
\[d(x)=m_1(x)\star(\ell(x),0)+ m_2(x)\star(s(x),\kappa_1t_1(x)+\kappa_2t_2(x)+\kappa_3t_3(x)+\kappa_4t_4(x)),\]
 for some  polynomials $m_1(x),m_2(x)\in R[x:\Theta].$
Let $m_1(x)=m'_0+m'_1x+\cdots+m'_\alpha x^\alpha$, where $m'_i=\kappa_1n'_{1,i}+\kappa_2n'_{2,i}+\kappa_3n'_{3,i}+\kappa_4n'_{4,i}$ for $ i=0,1,\ldots,\alpha.$ Then $m_1(x)$ can be  written as 
\begin{align*}
m_1(x)&=\kappa_1(n'_{1,0}+n'_{1,1}x+\cdots+n'_{1,\alpha} x^\alpha)+\kappa_2(n'_{2,0}+n'_{2,1}x+\cdots+n'_{2,\alpha}x^\alpha)\\
     &\hspace{1cm}+\kappa_3(n'_{3,0}+n'_{3,1}x+\cdots+n'_{3,\alpha} x^\alpha)+\kappa_4(n'_{4,0}+n'_{4,1}x+\cdots+n'_{4,\alpha} x^\alpha)\\
&=\kappa_1n_1'(x)+\kappa_2n_2'(x)+\kappa_3n_3'(x)+\kappa_4n_4'(x),
\end{align*}

where $n_i'(x)=n'_{i,0}+n'_{i,1}x+\cdots+n'_{i,\alpha} x^\alpha$ for $i=1,2,3,4.$ Now,
\begin{align*}
    m_1(x)\star(\ell(x),0)
    &= (\kappa_1n_1'(x)+\kappa_2n_2'(x)+\kappa_3n_3'(x)+\kappa_4n_4'(x))\star (\ell(x),0)\\
&=n_1'(x)\star(\eta(\kappa_1)\ell(x),0)+n_2'(x)\star(\eta(\kappa_2)\ell(x),0)\\
&\hspace{1cm}+n_3'(x)\star(\eta(\kappa_3)\ell(x),0)+n_4'(x)\star(\eta(\kappa_4)\ell(x),0)\\
    &= n_1'(x)\star(\ell(x),0).
    \end{align*}  
     If  $\deg(n_1'(x))< \deg(f(x))$ implies $
    n_1'(x)\star(\ell(x),0)\in \mbox{Span}(G_1)$. In other case, by division algorithm, we have   $n_1'(x)= \mu(x)f(x)+\nu(x)$  such that  $\nu(x)=0$ or $\deg(\nu(x))< \deg(f(x))$, for  $\mu(x), \nu(x)\in \mathbb F_q[x:\theta].$ Therefore, we get 
    \begin{align*}
      n_1'(x)\star(\ell(x),0)
      &= (\mu(x)f(x)+\nu(x))\star(\ell(x),0)\\
      &= \mu(x)f(x)\star(\ell(x),0)+\nu(x)\star(\ell(x),0)\\
      &= 0+\nu(x)\star(\ell(x),0).
      \end{align*}
    Thus, we conclude that     $n_1'(x)\star(\ell(x),0)\in \mbox{Span}(G_1)$. Now, it remains  to prove that 
     \[m_2(x)\star(s(x),\kappa_1t_1(x)+\kappa_2t_2(x)+\kappa_3t_3(x)+\kappa_4t_4(x))\in\mbox{Span}(G_1\cup G_2\cup G_3\cup G_4\cup G_5).\]
     Let $m_2(x)= m''_0+m''_1x+\cdots+m''_\beta x^\beta$, where $m''_i= \kappa_1n''_{1,i}+\kappa_2n''_{2,i}+\kappa_3n''_{3,i}+\kappa_4n''_{4,i} 
    ~\mbox{for}~i=0,1,\cdots,\beta.$  So,   
\begin{align*}
  m_2(x)&= \kappa_1(n''_{1,0}+n''_{1,1}x+\cdots+n''_{1,\beta} x^\beta)+\kappa_2(n''_{2,0}+n''_{2,1}x+\cdots+n''_{2,\beta}x^\beta)\\
  &+\kappa_3(n''_{3,0}+n''_{3,1}x+\cdots+n''_{3,\beta} x^\beta)+\kappa_4(n''_{4,0}+n''_{4,0}x+\cdots+n''_{4,1} x^\beta)\\
  &= \kappa_1n_1''(x)+\kappa_2n_2''(x)+\kappa_3n_3''(x)+\kappa_4n_4''(x),
\end{align*}
where $n_i''(x)=n''_{i,0}+n''_{i,1}x+\cdots+n''_{i,\alpha} x^\alpha$ for $i=1,2,3,4.$
Now,
     \begin{align*}
    &m_2(x)\star(s(x),\kappa_1t_1(x)+\kappa_2t_2(x)+\kappa_3t_3(x)+\kappa_4t_4(x))\\
    &= (\kappa_1n_1''(x)+\kappa_2n_2''(x)+\kappa_3n_3''(x)+\kappa_4n_4''(x))\star (s(x),\kappa_1t_1(x)+\kappa_2t_2(x)+\kappa_3t_3(x)+\kappa_4t_4(x))\\
    &=n_1''(x)\star(s(x),\kappa_1t_1(x))+n_2''(x)\star(0,\kappa_2t_2(x))+n_3''(x)\star(0,\kappa_3t_3(x))+n_4''(x)\star(0,\kappa_4t_4(x)).
    \end{align*}
 If $\deg(n_2''(x))< \deg(h_2(x))$ implies 
$n_2''(x)\star(0,\kappa_2t_2(x))\in \mbox{Span}(G_3)$. In other case, by division algorithm, we have   $n_2''(x)= q_2(x)h_2(x)+r_2(x)$  such that  $r_2(x)=0$ or $\deg(r_2(x))< \deg(h_2(x)),$ for  $q_2(x), r_2(x)\in R[x:\Theta]$.  Therefore, we get 
  \begin{align*}
      n_2''(x)\star(0,\kappa_2t_2(x))
      &= (q_2(x)h_2(x)+r_2(x))\star(0,\kappa_2t_2(x))\\
      &= q_2(x)h_2(x)\star (0,\kappa_2t_2(x))+r_2(x)\star(0,\kappa_2t_2(x))\\
      &= 0+r_2(x)\star(0,\kappa_2t_2(x)).
 \end{align*}
 It follows that    $n_2''(x)\star(0,\kappa_2t_2(x))\in \mbox{Span}(G_3)$. In a similar manner, we can prove  that  $n_3''(x)\star(0,\kappa_3t_3(x))\in \mbox{~Span}(G_4)\mbox {~and~} n_4''(x)\star(0,\kappa_4t_4(x))\in \mbox{~Span} (G_5)$.\par
       \vskip 3pt
     Now, consider $n_1''(x)\star (s(x),\kappa_1t_1(x)),$ if   $\deg(n_1''(x))< \deg(h_1(x))$ implies $n_1''(x)\star(s(x),\kappa_1t_1(x))\in \mbox{Span}(G_2)$. In other case, by division algorithm, we have $n_1''(x)= q_1(x)h_1(x)+r_1(x)$  such that  $r_1(x)=0$ or $\deg(r_1(x))< \deg(h_1(x))$, for $q_1(x), r_1(x)\in  R[x:\theta]$. Therefore, we get 
    \begin{align*}
      n_1''(x)\star(s(x),\kappa_1t_1(x))
      &=(q_1(x)h_1(x))\star(s(x),\kappa_1t_1(x))+r_1(x)\star(s(x),\kappa_1t_1(x))\\
      &= q_1(x)\star((h_1(x)\star s(x)),0)+r_1(x)\star(s(x),\kappa_1t_1(x)),
      \end{align*}
and  $r_1(x)\star(s(x),\kappa_1t_1(x))\in\mbox{~Span}(G_2).$ From Lemma \ref{lem-3.13}, we have  $\ell(x)\mid_rh_1(x)s(x)$ which implies,   $q_1(x)\star((h_1(x)\star s(x)),0)\in \mbox {~Span}(G_1).$ Thus, we get   $n_1''(x)\star (s(x),\kappa_1t_1(x))\in \mbox{~Span}(G_1\cup G_2).$ Hence, we have     $d(x)\in \mbox{~Span}(G_1\cup G_2\cup G_3\cup G_4\cup G_5).$ It is worth noting that the elements in   span of $G_1\cup G_2\cup G_3\cup G_4\cup G_5$ are linearly independent as $R[x:\Theta]$-submodule.  This implies  that   $G= \bigcup\limits_{i=1}^{5} G_i $ forms a minimal spanning  set of $\mathfrak{D}$. Moreover, $|\mathfrak{D}|= q^{\deg(f(x))+\sum_{i=1}^4\deg(h_i(x))}$. 
     
\end{proof}
 Now, we provide a detailed example to elaborate the previous result.

\begin{example}\em
Let $q=27, r=3=s$, $R= \mathbb F_{27}+u\mathbb F_{27}+v\mathbb F_{27}+uv \mathbb F_{27},$ where $\mathbb F_{27}=\mathbb F_3[\omega]$ with $\omega^3+2\omega+2=0$ and  Frobenius automorphism  $\theta(\alpha)= \alpha^3$ for all $\alpha\in \mathbb F_{27}$. We   have order of $\theta$, $\Theta $ is  $3$ such that  $\ord(\theta)\mid r$,  $\ord(\Theta)\mid s$.   Since $\mathbb F_{27}[x:\theta]$ is not a UFD, so the polynomial $x^{3}-1$ has  more than one factorizations. Consider  $\mathfrak{D} = \langle(\ell(x), 0),(s(x), t(x))\rangle$, it   is a  $(\theta, \Theta)$-cyclic code of  block length $(3, 3)$ over $\mathbb{F}_{27}R$, where 
\begin{eqnarray*}
\ell(x)&=& x+\omega^{17},\\
t_1(x)&=& t_2(x)= x+\omega^5,\\
t_3(x)&=& t_4(x)= x^2+\omega^6x+\omega^8,\\
s(x)&=&\omega.
\end{eqnarray*}
Next, we have 
\begin{eqnarray*}
f(x)&=& x^2+\omega^6x+\omega^{8},\\
h_1(x)&=& h_2(x)= x^2+\omega^{18}x+\omega^{22},\\
h_3(x)&=& h_4(x)= x+\omega^{17}.
\end{eqnarray*}
Now, $\mathfrak{D}$ possesses the generator matrix $G$ as shown below,
 \[G=
 \left( {\begin{array}{cccccc}
\omega^{17}&1&0&0&0&0\\
0&\omega^{25}&1&0&0&0\\
\omega&0&0&\omega^5\kappa_1&\kappa_1&0\\
0&\omega^3&0&0&\omega^{15}\kappa_1&\kappa_1\\
0&0&0&\omega^5\kappa_2&\kappa_2&0\\
0&0&0&0&\omega^{15}\kappa_2&\kappa_2\\
0&0&0&\omega^8\kappa_3&\omega^6\kappa_3&\kappa_3\\
0&0&0&\omega^8\kappa_4&\omega^6\kappa_4&\kappa_4\\
\end{array} } \right),\]
and 
according to Theorem \ref{theo-4.12}, $\mathfrak{D}$ has  $= 27^8$ codewords.
\end{example}
In Table \ref{taboptimal-1}, we provide some examples to demonstrate our findings.  We  have obtained some near-optimal codes, marked as $**$,  and optimal codes, marked as $*$ over the field $\mathbb F_9,$ where $\mathbb F_9= \mathbb F_3[\omega]$  with $\omega^2+\omega+1=0$. In rest of Table $1$, we use Frobenius automorphism over the  field $\mathbb F_9$ defined as $\theta(\alpha)=\alpha^3$ for all $\alpha\in \mathbb F_9$ and the coefficients of generator polynomials set in the ascending order for instance, $(\omega^7,\omega^3,1)$ represents the polynomial $\omega^7+\omega^3x+x^2.$   
\begin{table}[ht]\label{tab-1}
\centering
\caption{ Near-optimal and optimal codes from  $\Theta$-cyclic code $\mathfrak C_s$ over  $ R$}
\label{taboptimal-1}
\vspace{.5 cm}
\begin{tabular}{|p{1cm}|p{1cm}|p{2.5cm}|p{2.5cm}|p{1.5cm}|p{1.5cm}|c|}
\hline
$q$&$n$ &$t_1(x)$ & $t_2(x)$ & $t_3(x)$ & $t_4(x)$ & $\Phi_1(\mathfrak C_s)$\\

\hline
9 & 4  & $(\omega^7,\omega^3,1)$           & $(1,1)$                      & $(\omega^5,1)$     & $(\omega)$    & $[16,12,4]^*$    \\
\hline
9 & 6  & $(\omega^2,\omega,1)$             & $(1,1)$                      & $(\omega^6,1)$     & $(1)$         & $[24,20,3]^{**}$ \\
\hline
9 & 6  & $(2,\omega^5,\omega^3,1)$         & $(\omega^2,1)$               & $(2,1)$            & $(\omega)$    & $[24,19,4]^*$    \\
\hline
9 & 8  & $(\omega^7,\omega^6,\omega,1)$    & $\begin{aligned}(&\omega^6,\omega^3,1,\omega^6,\\ &\omega^7,1)\end{aligned}$ & $(\omega^7,0,1)$ & $(\omega^3,2,1)$ & $[32,19,8]^{**}$ \\
\hline
9 & 8  & $(\omega^7,\omega^6,\omega,1)$    & $(\omega,\omega^3,2,\omega^6,1)$             & $(1,\omega^3,1)$ & $(\omega,1,1)$  & $[32,21,7]^{**}$ \\
\hline
9 & 8  & $(\omega^3,\omega^7,\omega^2,1)$  & $(\omega^5,\omega^7,1)$                      & $(\omega^7,\omega,1)$ & $(1,1)$ & $[32,24,6]^*$    \\
\hline
9 & 8  & $(\omega^7,\omega^5,2,1)$         & $(\omega^7,1)$                                & $(\omega,1)$     & $(1)$         & $[32,27,4]^*$    \\
\hline
9 & 12 & $(2,\omega^6,\omega,1)$           & $(\omega^2,1)$                                & $(\omega^3,1)$   & $(1)$         & $[48,43,3]^{**}$ \\
\hline
9 & 12 & $(\omega^5,\omega^6,\omega^6,2,1)$ & $(\omega^2,1)$                               & $(\omega^3,1)$   & $(1)$         & $[48,42,4]^{**}$ \\
\hline
\end{tabular}
\end {table}

\section{Duality of  $(\theta, \Theta)$-cyclic codes over $\mathbb F_qR$}  
This section focuses on the  algebraic  structural  characteristics of the dual of   $(\theta, \Theta)$-cyclic code of block length  $(r, s)$ over $\mathbb{F}_qR$. Moreover, we determine the relation between  generator polynomials of $(\theta, \Theta)$-cyclic codes   and their duals.\par
 \vspace{.3cm}
In preliminaries section,  it has been   established   that the dual of $(\theta , \Theta )$-cyclic code $\mathfrak D$ of block length $(r,s )$ over $\mathbb F_qR$ is  again a $(\theta, \Theta)$-cyclic code. Therefore,  
\[\mathfrak{D}^\perp =\langle(\bar\ell(x),0), (\bar s(x),\bar t(x)\rangle,\] where  $\bar\ell(x)$ and $\bar s(x)\in \frac{\mathbb{F}_q[x:\theta]}{\langle x^r-1\rangle}$, 
 $\bar t(x)\in \frac{R[x:\Theta]}{\langle x^s-1\rangle}$, and  $\deg(\bar\ell(x)) <\deg(\bar s(x)).$\par
\vskip 3pt
Define $\Gamma_n(x)=\sum_{i=0}^{n-1}x^i$.
\begin{lemma}\cite{BCT18}
    Suppose that  ${\mathfrak m, \mathfrak n}\in \mathbb{N}.$ Then\[x^{ {\mathfrak m\mathfrak n}-1} = (x^{\mathfrak m}-1)\Gamma_{\mathfrak n}(x^{\mathfrak m}) = \Gamma_{\mathfrak n}(x^{m})(x^{\mathfrak m}-1).\]
\end{lemma}
\begin{definition}
  Let a non-zero polynomial  $\gamma(x)= a_0+a_1x+a_2x^2+\cdots+a_tx^t$  in $\mathbb F_q[x:\theta]$. Then    
\[\gamma^\dag(x)=a_t+\theta(a_{t-1})x+\theta^2(a_{t-2})x^2+\cdots +\theta^t(a_0)x^t,\] is called the  reciprocal polynomial of $\gamma(x).$
Alternatively, we  can express the reciprocal  as $\gamma^\dag(x)=\sum_{i=0}^{t}\theta^i(a_{t-i})x^i$.
\end{definition}

\begin{lemma} \cite [Lemma 5]{CLU09}
Suppose      $\Psi_1:\mathbb F_q[x:\theta]\rightarrow \mathbb F_q[x:\theta],$   given by

\[\sum_{k=0}^ rg_kx^k\to \sum_{k=0}^r\theta(g_k)x^k.\] 
  Then    $\Psi_1$ is  a  ring homomorphism.
\end{lemma}
\begin{proposition}\label{lem-5.5}
  Suppose  $\mu(x)$, $\nu(x)\in \mathbb F_q[x:\theta]$ are  any two  polynomials.  Then 
   \begin{enumerate}
     \item 
         If $\deg(\mu(x))\geq$ $\deg(\nu(x)),$ then $(\mu(x)+\nu(x))^\dag=\mu^\dag(x)+x^{\deg(\mu(x))-\deg (\nu(x))}\nu^\dag(x).$
         
        \item $(\mu(x)^\dag)^\dag=\Psi_1^r(\mu(x)) $, where $\deg(\mu(x))=r$.
        \item$(\mu(x)\nu(x))^\dag=\Psi_1^{\deg \mu(x)}(\nu^\dag(x))\mu^\dag(x).$
   
    \end{enumerate}
    \end{proposition}
    \begin{proof}
       
      Arguments such as those used in proof of Lemma 2.8 in  \cite{HRS21} can be used to proof the proposition. 
    \end{proof}
    
\begin{lemma}\cite[Lemma 5]{CLU09}
       Suppose $\Psi_2$ is a map given as follows:      $$\Psi_2: R[x:\Theta]\to R[x:\Theta]$$    defined by  
\[\sum_{j=0}^s b_jx^j\rightarrow \sum_{j=0}^s\Theta(b_j)x^j,\] 
where $b_j\in  R.$ Then  $\Psi_2$ is  a ring homomorphism.
 \end{lemma}

 \begin{proposition}
  Suppose  $\mu_1(x)$, $\nu_1(x)\in \mathbb R[x:\Theta]$ are  any two  polynomials.  Then 
   \begin{enumerate}
    \item$(\mu_1(x)\nu_1(x))^\dag=\Psi_2^{\deg \mu_1(x)}(\nu_1^\dag(x))\mu_1^\dag(x).$
    \item $(\mu_1(x)^\dag)^\dag=\Psi_2^s(\mu(x)) $, where $\deg(\mu_1(x))=s$.
    \item 
         If $\deg(\mu_1(x))\geq$ $\deg(\nu_1(x)),$ then $(\mu_1(x)+\nu_1(x))^\dag=\mu_1^\dag(x)+x^{\deg(\mu_1(x))-\deg (\nu_1(x))}\nu_1^\dag(x).$
        
    \end{enumerate}
    \end{proposition}
    \begin{proof}
       Arguments such as those used in proof of Lemma 5 in  \cite{HRS21} can be used to proof the proposition.  
    \end{proof}
     We assume $ \mathfrak m = r\times s$ throughout this paper.
\begin{definition}\label{def-5.6}
Let  $ d(x) = (g(x), b(x)),~d' (x)= (g'(x)), b'(x))\in \mathfrak{R}_{r,s}$. Define a map    \[ o : \mathfrak{R}_{r,s}\times \mathfrak R_{r,s}\rightarrow \frac{R[x:\Theta]}{\langle x^{\mathfrak m}-1\rangle},\]  
by: 
 \begin{align*}
     o(d(x),~ d' (x))& = \kappa_1g(x)\Psi_1^{\mathfrak m - \deg(g'(x))}(g'^\dag(x))x^{\mathfrak m-1-\deg(g'(x))}\Gamma_\frac{\mathfrak m}{r}(x^r)\\
    &\hspace{2.7cm}+ b(x)\Psi_2^{\mathfrak m-\deg(b'(x))}(b'^\dag(x))x^{\mathfrak m-1-\deg(b'(x)}\Gamma_{\frac{\mathfrak m}{s}}(x^s)\pmod{(x^{ \mathfrak m}-1)}.
 \end{align*}
  Clearly, the  map $``o"$ is bilinear.  We denote $o(d(x), d' (x))$ by $d(x)~o~ d' (x)$ for convenience.
 
 \end{definition}

\begin{proposition}\label{prop-5.3}  
Let  $\textit{\textbf d}$ and  $\textit{\textbf d}^\prime$  be   two vectors in $\mathbb{F}_q^r\times R^s$ and their  associated polynomials be $d(x)=(g(x),b(x))$ and $d' (x)=(g'(x),b'(x)) \in   \mathfrak{R}_{r,s}$. Then $\textit{\textbf d}$ and all its $(\theta, \Theta)$-cyclic shifts are orthogonal to $\textit{\textbf d}'$   if and only if $d(x)~ o ~d'(x) \equiv 0\pmod{x^\mathfrak{m}-1}.$
\end{proposition}

\begin{proof}
Suppose $\textit{\textbf d}=(g_0,g_1,\ldots, g_{r-1} b_0,b_1,\ldots,b_{s-1})$ and $\textit{\textbf d}^\prime =(g'_0,g'_1,\ldots, g'_{r-1}, b'_0,b'_1,\ldots,b'_{s-1})\in\mathbb{F}_q^r\times R^s$.  For $0\le i\le \mathfrak m-1$, the $i^{th}$  $(\theta,\Theta)$-cyclic shift of $\textit{\textbf d} $ is   represented by the vector   $\textit{\textbf d}^{(i)}=(\theta^i(g_{0-i}),\theta^i(g_{1-i}),\ldots,
\theta^i(g_{r-1-i}), \Theta^i(b_{0-i}),\Theta^i(b_{1-i}),\ldots,
\Theta^i(b_{s-1-i}))$.  Now,  $\textit{\textbf d}^{(i)} \cdot \textit{\textbf d}^\prime \equiv 0\pmod{x^\mathfrak m-1}$ if and only if\[\kappa_1\sum_{j=0}^{r-1}\theta^i(g_{j-i})g'_j+\sum_{\nu=0}^{s-1}\Theta^i(b'_{\nu-i})b'_\nu=0.\] 
Let $S_i=\kappa_1\sum_{j=0}^{r-1}\theta^i(g_{j-i})g'_j+\sum_{\nu=0}^{s-1}\Theta^i(b'_{\nu-i})b'_\nu.$ Then we have 
\begin{align*}
d(x)\ o\ d'(x)&= \kappa_1\left[\sum_{\mu=0}^{r-1}\sum_{j=\mu}^{r-1}g_{j-\mu}\theta^{\mathfrak m-\mu}(g'_j)x^{\mathfrak m-1-\mu}+ \sum_{\mu=1}^{r-1}\sum_{j=\mu}^{r-1}g_j\theta^\mu(g'_{j-\mu})x^{\mathfrak m-1+\mu}\right]\Gamma_{\frac{\mathfrak m}{r}}(x^r)\\
&\hspace{2cm}+\left[\sum_{k=0}^{s-1}\sum_{\nu=k}^{s-1}b_{\nu-k}\Theta^{\mathfrak m-k}(b'_\nu)x^{\mathfrak m-1-k}+\sum_{k=1}^{s-1}\sum_{\nu=k}^{s-1}b_{k}\Theta^{k}(b'_{\nu-k})x^{ \mathfrak m-1+k}\right]\Gamma_{\frac{\mathfrak m}{s}}(x^s)\\
&= \sum_{i=0}^{\mathfrak  m-1}\rho_{(\theta,\Theta)}^{\mathfrak m-i}\ S_ix^{\mathfrak  m-1-i}\pmod{x^{\mathfrak  m}-1}.
\end{align*}
 Hence, $d (x)~ o~ d'(x)\equiv 0\pmod{x^\mathfrak{m}-1}$ if and only if $S_i=0$ for $i=0,1,\ldots,\mathfrak m-1.$ This proves the proposition.
\end{proof}
 
 The following proposition is  used to determine   the generator polynomials for $\mathfrak{D}^\perp$.

\begin{proposition}\label{lem-5.8}
     Let   $d(x)=(g(x), b(x)),~d^\prime (x)=(g'(x), b'(x))\in \mathfrak{R}_{r,s}$ and  $d(x)\ o\ d^\prime (x)\equiv 0\pmod{x^\mathfrak m-1}.$ Then 
    \begin{enumerate}
       \item  $g(x)\Psi_1^{\mathfrak m-\deg(g'(x))}(g'^\dag(x))\equiv 0 \pmod{x^r-1},$  if $b(x)=0$ or $b'(x)=0.$ 
        \item 
         $b(x)\Psi_2^{\mathfrak m-\deg(b'(x))}(b'^\dag(x))\equiv 0\pmod{x^s-1},$  if   $g(x)=0$ or $g'(x)=0.$
    \end{enumerate}
\end{proposition}

\begin{proof}
    Assume that    $b(x)=0$ or $b'(x)=0$, then 
    \[d(x)\ o\ d'(x)=\kappa_1g(x)\Psi_1^{\mathfrak m-\deg(g'(x))}(g'^\dag(x))x^{\mathfrak m-1-\deg(g'(x))}\Gamma_\frac{\mathfrak m}{r}(x^r)+0\equiv 0\pmod {x^{\mathfrak m}-1},\] i.e.,  \[\kappa_1g(x)\Psi_1^{\mathfrak m-\deg(g'(x))}(g'^\dag(x))x^{{\mathfrak m}-1-\deg(g'(x))}\Gamma_\frac{\mathfrak m}{r}(x^r)=\kappa_1\lambda(x)(x^\mathfrak m-1),\] form some $\lambda(x)\in \mathbb F_q[x:\theta]$.
    Since $\Gamma_\frac{\mathfrak m}{r}(x^r)=\frac{x^\mathfrak m-1}{x^r-1}$ and $(x^{\mathfrak m}-1)(x^r-1)=(x^r-1)(x^{\mathfrak m}-1)$. Hence,
    \[ \kappa_1g(x)\Psi_1^{\mathfrak m-\deg(g'(x))}(g'^\dag(x))x^{\mathfrak m}=\kappa_1\lambda(x)x^{(\deg(g'(x))+1)}(x^r-1).\]
    Consequently, \[g(x)\Psi_1^{\ \mathfrak m-\deg(g'(x))}(g'^\dag(x))\equiv0\pmod{x^r-1}.\]
   
   The second case  can be proved using the similar argument. 
\end{proof}
  
\begin{proposition}\label{prop-5.9}
     Let $\mathfrak{D} = \langle(\ell(x), 0),(s(x), t(x))\rangle$, where $t(x) = \kappa_1t_1(x) + \kappa_2t_2(x)+\kappa_3t_3(x)+\kappa_4t_4(x)$. Then
     \begin{align*}
	& |\mathfrak C_r|=q^{r-\deg(\gcd_r(\ell(x),s(x)))}, ~| \mathfrak C_s|=q^{4s -\sum_{i=1}^4\deg(t_i(x))},\\
 &|\kappa_1\mathfrak C_s|=q^{s-\deg(t_1(x))},~ |\kappa_2\mathfrak C_s|=q^{s-\deg(t_2(x))},\\
& |\kappa_3\mathfrak C_s|=q^{s-\deg(t_3(x))}, ~|\kappa_4\mathfrak C_s|=q^{s-\deg(t_4(x))},\\
	&|(\mathfrak C_r)^\perp |= q^{\deg(\gcd_r(\ell(x),s(x)))},~ |( \mathfrak C_s)^\perp|=q^{ \sum_{i=1}^4\deg(t_i(x))},\\
	&|(\mathfrak  C^\perp)_r |= q^{\deg(\ell(x))},~ |( \mathfrak C^\perp)_s|=q^{ \sum_{i=1}^4\deg(t_i(x))+\deg(\ell(x))-\deg({\gcd}_r(\ell(x),s(x)))},\\
& |\kappa_1(\mathfrak C_s)^\perp|=q^{\deg(t_1(x))},~|\kappa_2(\mathfrak C_s)^\perp|=q^{\deg(t_2(x))},\\
& |\kappa_3(\mathfrak C_s)^\perp|=q^{\deg(t_3(x))},~|\kappa_4(\mathfrak C_s)^\perp|=q^{\deg(t_4(x))},\\
& |\kappa_1(\mathfrak C^\perp)_s|=q^{\deg(t_1(x)+\deg(\ell(x)-\deg(\gcd_r(\ell(x),s(x)))},~|\kappa_2(\mathfrak C^\perp)_s|=q^{\deg(t_2(x))},\\
& |\kappa_3(\mathfrak C^\perp)_s|=q^{\deg(t_3(x))},~ |\kappa_4(\mathfrak C^\perp)_s|=q^{\deg(t_4(x))}.
\end{align*}
 \end{proposition}
 \begin{proof} Since we have  proved in Theorem \ref{theo-4.7} that $\mathfrak C_r=\langle \gcd_r(\ell(x),s(x))\rangle$ and $\mathfrak C_s=\langle t(x)\rangle=\langle \kappa_1t_1(x) + \kappa_2t_2(x)+\kappa_3t_3(x)+\kappa_4t_4(x)\rangle$.  Therefore,  $|\mathfrak C_r|=q^{r-\deg(\gcd_r(\ell(x), s(x))}$ and $|\mathfrak C_s|= q^{4s-\sum_{i=1}^4\deg(t_i(x))}$. We can prove  other results  in a similar way.

 \end{proof}
 In  next proposition, we 
  determine the degree of  generator polynomials for $\mathfrak{D}^\perp$.
 
 \begin{proposition}\label{prop-5.10}
     Let $\mathfrak{D}= \langle(\ell(x), 0), (s(x),t(x))\rangle$, where $t(x)= \kappa_1 t_1(x)+\kappa_2 t_2(x)+\kappa_3 t_3(x)+\kappa_4 t_4(x) $  and $\mathfrak{D}^\perp = \langle(\bar\ell(x), 0), (\bar s(x), \bar t(x))\rangle$. Then
     \begin{align*}
         & \deg(\bar\ell(x))=r-\deg({\gcd}_r(\ell(x), s(x)),\\
         &\deg(\bar t_1(x))=s-\deg(t_1(x))-\deg(\ell(x))+\deg({\gcd}_r(s(x),\ell(x))),\\
         &\deg(\bar t_j(x))=s-\deg(t_j(x)),~ \mbox{for}~j=2,3,4.\\
         \end{align*}
  \end{proposition}
  \begin{proof}
     It can be easily established that $(\mathfrak C_r)^\perp=\langle \bar\ell(x)\rangle$.  So $|(\mathfrak C_r)^\perp|=q^{r-\deg(\bar\ell(x))}.$ By Proposition \ref{prop-5.9}, we get  $|(\mathfrak C_r)^\perp|=q^{\deg(\gcd_r(\ell(x),s(x)))}$. Hence, $\deg(\bar\ell(x))=r-\deg(\gcd_r(\ell(x),s(x))).$
      Again, it easily follows that   $\kappa_1(\mathfrak C^\perp)_s=\langle \kappa_1 \bar t_1(x) \rangle$. Therefore, $|\kappa_1(\mathfrak C^\perp)_s|=q^{s-\deg(\bar t_1(x))}.$ Now,  by   Proposition \ref{prop-5.9}, we have    $|\kappa_1(\mathfrak C^\perp)_s|=q^{\deg(t_1(x))+\deg(\ell(x))-\deg(\gcd_r(\ell(x),s(x)))}$. This shows that   \linebreak 
 $\deg(\bar t_1(x))=s-{\deg(t_1(x))-\deg(\ell(x))+\deg(\gcd_r(\ell(x),s(x)))}.$ 
      Similarly,  the remaining parts can also be proved.
  \end{proof}
 In the  next theorem, we  calculate the generator polynomial for $\mathfrak D^\perp.$

 \begin{theorem}\label{theo-5.11}
     Let  $\mathfrak{D} = \langle(\ell(x), 0),(s(x), t(x))\rangle$.  Then \[\mathfrak{D}^\perp =\langle(\bar\ell(x),0), (\bar s(x),\bar t(x)\rangle ,\] such that
     \begin{enumerate}
         \item $\bar\ell(x) = \frac{x^r-1}{\gcd_l\left(\Psi_1^{ {\mathfrak m}-\deg(\ell(x))}(\ell^\dag(x)),
         \Psi_2^{ {\mathfrak m}-\deg(s(x))}(s^\dag(x))\right)},$
         \item $\bar s(x) = \frac{\lambda(x)(x^r-1)}{\Psi_1^{ {\mathfrak m}-\deg(\ell(x))}(\ell^\dag(x))}$, where $\lambda(x)\in {\mathbb{F}_q[x:\theta]},$
         \item $\bar t_1(x)  = \frac{x^s-1}{\Psi_2^{{\mathfrak m}-\deg\left(\frac{\ell(x)t_1(x)}{\gcd_l(\ell(x), s(x)}\right)}\left(\frac{\ell(x)t_1(x)}{\gcd_l(\ell(x), s(x))}\right)^\dag}$,
         \item $\bar t_2(x) = \frac{x^s-1}{\Psi_2^{{\mathfrak m}-\deg t_2(x)}(t_2^\dag(x))},$
         \item $\bar t_3(x) = \frac{x^s-1}{\Psi_2^{{\mathfrak m}-\deg t_3(x)}(t_3^\dag(x))},$
         \item $\bar t_4(x) = \frac{x^s-1}{\Psi_2^{{\mathfrak m}-\deg t_4(x)}(t_4^\dag(x))}.$
     \end{enumerate}
  
 \end{theorem}
\begin{proof} (1) Since $(\bar\ell(x), 0 )\in \mathfrak{D}^\perp$. So  $(\bar\ell(x), 0 ) ~o~ (\ell(x), 0) \equiv 0\pmod{x^\mathfrak {m}-1}$ and \linebreak $(\bar\ell(x), 0 )~o~(s(x), t(x)) \equiv 0\pmod{x^\mathfrak {m}-1}$.
Now, by  Proposition  \ref{lem-5.8}, we obtain\[\bar\ell(x)\Psi_1^{{\mathfrak m}-\deg(\ell(x))}(\ell^\dag (x)) \equiv 0 \pmod
{x^r-1}\] and \[\bar\ell(x)\Psi_1^{ {\mathfrak m}-\deg(s(x))}(s^\dag (x)) \equiv 0 \pmod{x^r-1}.\]  Since  $(x^r-1)\in Z(\mathbb{F}_q[x:\theta])$, by    Lemma \ref{lem-2.10}, we obtain \[\bar\ell(x){\gcd}_l\left(\Psi_1^{{\mathfrak m}-\deg(\ell(x)}(\ell^\dag(x)), \Psi_1^{{\mathfrak m}-\deg(s(x))}(s^\dag(x))\right)\equiv 0\pmod{x^r-1},\] i.e.,   \[\bar\ell(x){\gcd}_l\left(\Psi_1^{{\mathfrak m}-\deg(\ell(x)}(\ell^\dag(x)), \Psi_1^{{\mathfrak m}-\deg(s(x))}(s^\dag(x))\right) =  \mu(x)(x^r-1),\] for $\mu(x)\in\mathbb{F}_q[x:\theta]$.  
As $\deg(\bar \ell(x)) = r-\deg({\gcd}_r(\ell(x), s(x))$, so   $\mu(x) = 1$. Hence, \[\bar\ell(x) = \frac{x^r-1}{\gcd_l\left(\Psi_1^{{\mathfrak m}-\deg(\ell(x))}(\ell^\dag(x)), \Psi_1^{{\mathfrak m}-deg(s(x))}(s^\dag(x)\right)}.\]
$(2)$  As $d'(x) =  (\bar s(x), \bar t(x))\in \mathfrak{D}^\perp,$ then   \[d'(x) o (\ell(x), 0) =  (\bar s(x), \bar t(x))~ o~ (\ell(x), 0),\] by the bilinearity of map $``o"$, we have \[\bar s(x)\Psi_1^{{\mathfrak m}-\deg(\ell(x))}(\ell^\dag(x)) \equiv 0 \pmod{x^r-1}.\] i.e.,  \[\bar s(x)\Psi_1^{{\mathfrak m}-\deg(\ell(x))}(\ell^\dag(x)) = \lambda(x)(x^r-1)\] for some   $\lambda(x)\in \mathbb{F}_q[x:\theta].$ It follows that  \[\bar s(x)= \frac{\lambda(x)(x^r-1)}{\Psi_1^{{\mathfrak m}-\deg(\ell(x))}(\ell^\dag(x))}.\] 
(3)  Since  $\kappa_1 \star (\bar s(x), \kappa_1\bar t_1(x)+\kappa_2\bar t_2(x)+\kappa_3\bar t_3(x)+\kappa_4\bar t_4(x)) = (\bar s(x), \kappa_1\bar t_1(x))\in \mathfrak{D}^\perp,$ then
\begin{align*}
&\kappa_1\left(\frac{\ell(x)}{\gcd_r(\ell(x), s(x))}\right)\star(s(x), \kappa_1t_1(x)+\kappa_2t_2(x)+\kappa_3t_3(x)+\kappa_4t_4(x))\\
&\hspace{5cm}- \frac{s(x)}{\gcd_r(\ell(x), s(x))}\star \left(\ell(x), 0\right)\\
&~~~~~~~~~~~~~~~~~~~~~~~~~~~~~~~=\left(0, \frac{\kappa_1\ell(x)t_1(x)}{\gcd_r(\ell(x), s(x))}\right)\in \mathfrak{D}.
\end{align*}
This implies, \[ (\bar s(x), \kappa_1\bar t_1(x))\ o\left(0, \kappa_1\frac{\ell(x)t_1(x)}{\gcd_r(\ell(x), s(x))}\right) = 0 .\] Using Lemma \ref{lem-5.8}, we get  \[\kappa_1\bar t_1(x)\Psi_2^{{\mathfrak m}-\deg\left(\frac{\ell(x)t_1(x)}{\gcd_r(\ell(x), s(x))}\right)}\left(\frac{\ell(x)t_1(x)}{\gcd_r(\ell(x), s(x))}\right)^\dag \equiv 0\pmod{x^s-1}.\] 
 Hence, $\kappa_1\bar t_1(x)\Psi_2^{{\mathfrak m}-\deg\left(\frac{\ell(x)t_1(x)}{\gcd_r(\ell(x), s(x))}\right)}\left(\frac{\ell(x)t_1(x)}{\gcd_r(\ell(x), s(x))}\right)^\dag = \kappa_1\gamma(x)(x^s-1),$ where $\gamma(x)\in  R[x:\Theta].$ 
 So, \[\bar t_1(x) = \frac{\gamma(x)(x^s-1)}{\Psi_2^{{\mathfrak m}-\deg\left(\frac{\ell(x)t_1(x)}{\gcd_r(\ell(x),s(x))}\right)}\left(\frac{\ell(x)t_1(x)}{\gcd_r(\ell(x), s(x))}\right)^\dag}.\]
 Now, Proposition \ref{prop-5.10},  implies that  $\deg(\bar t_1(x)) = s-\deg (t_1(x))-\deg(\ell(x))+\deg(\gcd_r(\ell(x), s(x))),$ so $\gamma(x) = 1.$\par
 \vskip 5pt
 
 Therefore, $\bar t_1(x) = \frac{x^s-1}{\Psi_2^{{\mathfrak m}-\deg\left(\frac{\ell(x)t_1(x)}{\gcd_r(\ell(x), s(x))}\right)}\left(\frac{\ell(x)t_1(x)}{\gcd_r(\ell(x), s(x))}\right)^\dag}$.\par 
 \vskip 5pt
\noindent(4)~~~ Since $\kappa_2\star(\bar s(x), \kappa_1\bar t_1(x)+\kappa_2\bar t_2(x)+\kappa_3\bar t_3(x)+\kappa_4\bar t_4(x)) = (0, \kappa_2\bar t_2(x))\in \mathfrak{D}^\perp $. Therefore, \[(0, \kappa_2\bar t_2(x))~ o~ (s(x), \kappa_1\bar t_1(x)+\kappa_2\bar t_2(x)+\kappa_3\bar t_3(x)+\kappa_4\bar t_4(x)) \equiv 0\pmod{x^\mathfrak m-1}.\]  Lemma \ref{lem-5.8}, implies that  \[\kappa_2\bar t_2(x)\Psi_2^{{\mathfrak m}-\deg(t_2(x))}(t_2^\dag(x)) \equiv 0 \pmod{x^s-1},\] i.e.,  \[\kappa_2\bar t_2(x)\Psi_1^{{\mathfrak m}-\deg t_2(x)}(t_2^\dag(x)) = \nu(x)(x^s-1)\] for  some $\nu(x)\in R[x:\Theta]$.   Next, by Proposition \ref{prop-5.10}, it follows that  $\deg(\bar t_2(x))= s-\deg(t_2(x))$. Therefore, $\nu(x) = 1.$ Consequently, we have    \[\bar t_2(x) = \frac{x^s-1}{\Psi_2^{{\mathfrak m}-\deg t_2(x)}(t_2^\dag(x))}.\] 
By using similar argument as above, we get
 \[\bar t_3(x) = \frac{x^s-1}{\Psi_2^{{\mathfrak m}-\deg t_3(x)}(t_3^\dag(x))},\] and 
         \[\bar t_4(x) = \frac{x^s-1}{\Psi_2^{{\mathfrak m}-\deg t_4(x)}(t_4^\dag(x))}.\]

 \end{proof}

\section{QECCs from $(\theta,\Theta)$-cyclic codes over $\mathbb{F}_qR$}
Quantum computers can potentially solve intricate problems at a significantly accelerated pace compared to classical computers. The field of quantum computing employs QECCs to shield quantum information from potential errors stemming from factors such as decoherence and quantum noise.\par
 Let $ \mathbb H$ be a $q$-dimensional Hilbert space over complex number $(\mathbb C)$  and $ \mathbb H^{\otimes n}$ denotes the $n$-fold tensor product of $ \mathbb H.$ A $q$-array QECC with parameters $[[n,k,d]]_q$ is a  $k$-dimensional subspace of $ \mathbb H^{\otimes n},$ where $n$ is length of code, $k$ is dimension of code and $d$ represents the Hamming distance of code. Moreover,  rate of an QECC $[[n,k,d]]_q$  is  defined as $\frac{k}{n}.$ An  $[[n_1, k_1, d_1]]_q$ QECC  is said to be better than any other   $[[n_2, k_2, d_2]]_q$  QECC if any one or both of the following holds:

\begin{enumerate}
  \item  if $d_1=d_2$, then $\frac{k_1}{n_1}> \frac{k_2}{n_2}$,
  \item  if $\frac{k_1}{n_1}= \frac{k_2}{n_2}$, then $d_1> d_2.$
\end{enumerate}
QECCs have various applications in quantum information theory. In classical computing, errors can be fixed using error-correcting codes. Quantum computing, however, is more sensitive to errors because quantum bits (qubits) can be affected by tiny disturbances. Quantum error-correcting codes (QECCs) are designed to protect quantum information from these errors, ensuring that computations are reliable.\par
These codes can detect and correct errors that occur during quantum operations. Without such mechanisms, quantum computers would be prone to errors that could render their results useless. Quantum codes also help in improving the efficiency and accuracy of quantum computations. They allow for more robust algorithms that can tolerate errors, making practical quantum computing more feasible. This is vital because quantum computers promise to solve problems that are currently intractable for classical computers.\par 
Quantum codes are fundamental for secure quantum communication. They enable transmitting quantum information over long distances without significant loss or corruption. This is essential for quantum cryptography, which offers theoretically unbreakable security. Apart from this, quantum codes have applications in quantum teleportation and Fault-tolerant quantum computing.\par
This section presents  a method to    construct  QECCs from separable  $(\theta,\Theta)$-cyclic code over $\mathbb F_qR$.
\begin{definition}\cite{KKKS06}(Quantum singleton bound)
Let  $\mathcal{Q}$ be  an $[[n,k,d]]_q$ QECC. Then  $2d \leq n - k + 2$. Additionally, $\mathcal{Q}$ is called a Maximum Distance Separable (MDS) QECC if $2d = n - k + 2$. 
\end{definition}

We now present some results which will be useful in constructing QECCs.
\begin{theorem}\label{theo-6.1}
    Let  $\mathfrak{D}$ be a   linear code of block length $(r,s)$ over $\mathbb{F}_qR$.  Then $\Phi(\mathfrak{D}^\perp) = \Phi(\mathfrak{D})^\perp.$ 
    Moreover, $\Phi(\mathfrak{D})$ is self-dual if and only if    $\mathfrak{D}$ is self-dual.
\end{theorem}
\begin{proof}
    Let $\textit{\textbf d}_1 =(g_0,g_1,\ldots,g_{r-1}, b_0,b_1,\ldots,b_{s-1})\in \mathfrak{D}$,
    $\textit{\textbf d}_2 =(g'_0,g'_1,\ldots,g'_{r-1}, b'_0,b'_1,\ldots,b'_{s-1})\in\mathfrak{D}^\perp$,
    where $b_i = \kappa_1n_{1,i}+\kappa_2n_{2,i}+\kappa_3n_{3,i}+\kappa_4n_{4,i}$ and $b'_i = \kappa_1n'_{1,i}+\kappa_2n'_{2,i}+\kappa_3n'_{3,i}+\kappa_4n'_{4,i}~\mbox{for}~ i =0,1,2,\ldots,s-1.$ Now,  
       \[\textit{\textbf d}_1\cdot\textit{\textbf d}_2 = \kappa_1\sum_{j=0}^{r-1} g_j g'_j + \sum_{i=0}^{s-1} b_i b'_i =0,\] which implies
    \begin{align*}
         \kappa_1\sum_{j=0}^{r-1} g_j g'_j +\sum_{i=0}^{s-1}(\kappa_1n_{1,i}+\kappa_2n_{2,i}+\kappa_3n_{3,i}+\kappa_4n_{4,i})(\kappa_1n'_{1,i}+\kappa_2n'_{2,i}+\kappa_3n'_{3,i}+\kappa_4n'_{4,i})&=0\\
         \kappa_1\sum_{j=o}^{r-1} g_i g'_i +\sum_{i=0}^{s-1}(\kappa_1n_{1,i}n'_{1,i}+\kappa_2n_{2,i}n'_{2,i}+\kappa_3n_{3,i}n'_{3,i}+\kappa_4n_{4,i}n'_{4,i})&=0.
     \end{align*}
    Now, by comparing the coefficients of   $\kappa_1$, $\kappa_2$, $\kappa_3$, and $\kappa_4$ on both the sides, we get
    \begin{align*}
    \sum_{j=0}^{r-1} g_jg'_j +\sum_{i=0}^{s-1}n_{1,i}n'_{1,i}&=0,\\
    \sum_{i=0}^{s-1}n_{2,i}n'_{2,i}&=0,\\
    \sum_{i=0}^{s-1}n_{3,i}n'_{3,i}&=0,\\
    \sum_{i=0}^{s-1}n_{4,i}n'_{4,i}&=0.
    \end{align*}
Next, we have   
\begin{align*}
\Phi(\textit{\textbf d}_1)\cdot\Phi(\textit{\textbf d}_2)&=\sum_{j=0}^{r-1} g_jg'_j + \sum_{i=0}^{s-1}b_iMM^Tb'_i,\\
    &=\sum_{j=0}^{r-1} g_jg'_j + \sum_{i=0}^{s-1}(n_{1,i}n'_{1,i}+n_{2,i}n'_{2,i}+n_{3,i}n'_{3,i}+n_{4,i}n'_{4,i}).
\end{align*}
Since $MM^T =  I_{4\times 4}.$
Thus, from the above equation, we have  $\sum_{j=0}^{r-1} g_jg'_j + \sum_{i=0}^{s-1}(n_{1,i}n'_{1,i}+n_{2,i}n'_{2,i}+n_{3,i}n'_{3,i}+n_{4,i}n'_{4,i}) = 0$. This implies that $\Phi(\textit{\textbf d}_2) \in \Phi(\mathfrak{D})^\perp$ for $\Phi(\textit{\textbf d}_1) \in \Phi(\mathfrak{D})$. Therefore, $\Phi(\mathfrak{D}^\perp) \subseteq \Phi(\mathfrak{D})^\perp$.
Since $\Phi$ is a bijection, so we have $|\Phi(\mathfrak{D}^\perp)| = |\Phi(\mathfrak{D})^\perp|$. Thus, we can conclude that $\Phi(\mathfrak{D}^\perp) = \Phi(\mathfrak{D})^\perp.$
Furthermore, if  $\mathfrak{D}$ is a self-dual code, then  $\Phi(\mathfrak{D}) = \Phi(\mathfrak{D}^\perp) = \Phi(\mathfrak{D})^\perp$. Consequently, we can conclude that $\Phi(\mathfrak{D})$ is also self-dual. 
\end{proof}
To construct QECCs in our setup, we present a well-known result, called CSS construction, discussed by Calderbank, Shor and Stean in 1998.
\begin{theorem}\label{theo-6.2} \cite[CSS construction]{CRSS98}
 Let   $\mathfrak D_{1}=[n,k_1,d_1]$ and $\mathfrak D_{2}=[n,k_2,d_2 ]$ be linear codes  over $\mathbb{F}_q$ with $\mathfrak D_{2}^\perp\subseteq \mathfrak D_{1}$.  Then  a $[[n,k_1+k_2-n,\mbox{min}\{d_1,d_2\}]]_q$  QECC   exists. Particularly, if  $\mathfrak D_{1}^\perp\subseteq \mathfrak D_{1}$, then a $[[n,2k_1-n,d_1 ]]_q$ QECC  can be constructed. 
    
\end{theorem}
Following theorems   present the dual-containing property of $\theta $-cyclic codes over $\mathbb F_q$, $\Theta$-cyclic codes over $R$ and separable $(\theta, \Theta)$-cyclic codes over $\mathbb F_qR$. Here, we assume $\ord(\theta)\mid r$ and $\ord(\Theta)\mid s$.
\begin{theorem}\label{theo-6.4}\cite{DBUBT12}
 Let   $\mathfrak C_r=\langle\ell(x)\rangle$ be  a $\theta$-cyclic code of length $r$ over  $\mathbb F_q.$ Then $\mathfrak C_r^\perp\subseteq \mathfrak C_r$  if and only if     $x^r-1\mid _r f^\dag(x)f(x)$,  where $x^r-1=f(x)\ell(x).$ 
\end{theorem}

\begin{theorem}\label{theo-6.5}
  Suppose  $\mathfrak C_s= \langle\kappa_1t_1+ \kappa_2t_2+\kappa_3t_3+\kappa_4t_4\rangle$ is  a $\Theta $-cyclic code of length $s$ over $R.$ Then    $\mathfrak C_s^\perp\subseteq\mathfrak C_s$ if and only if   $ x^s-1\mid_r h^{\dag}_i(x)h_i(x)$, where  $h_i(x)t_i(x)=x^s-1$ for $i=1,2,3,4.$
\end{theorem}

\begin{proof}
    We know that  $\mathfrak C_s=\kappa_1\mathfrak C_{s,1}\oplus\kappa_2\mathfrak C_{s,2}\oplus\kappa_3\mathfrak C_{s,3}\oplus\kappa_4\mathfrak C_{s,4},$  is $\Theta$-cyclic code of length $s$ over $R$  if and only if each $\mathfrak C_{s,i}$ is  $\theta$-cyclic code of length $s$ over $\mathbb F_q$. Also, $\mathfrak C_{s}^{\perp}\subseteq \mathfrak C_s$ if and only if $\mathfrak C_{s,i}^{\perp}\subseteq \mathfrak C_{s,i}.$ Further   $\mathfrak C_{s,i}=\langle t_i(x)\rangle$ such that   $x^s-1=h_i(x)t_i(x)$ for $i=1,2,3,4$ .  Now, Theorem \ref{theo-6.4} infer that    $\mathfrak C_{s,i}^{\perp}\subseteq \mathfrak C_{s,i}$ if and only if $x^s-1~|_r~ h^\dag_i(x)h_i(x)$.
\end{proof}

\begin{theorem}\label{theo-6.6}
Suppose  $\mathfrak{D}$ is  a separable $(\theta, \Theta)$-cyclic code over  $\mathbb F_qR$. Then $\mathfrak{D}^\perp\subseteq \mathfrak{D}$ if and only if $\mathfrak C^\perp_r\subseteq \mathfrak C_r$ and $\mathfrak C^\perp_s\subseteq \mathfrak C_s$.
\end{theorem}
\begin{proof}
    Assume that  $\mathfrak D^\perp\subseteq \mathfrak D= \mathfrak C_r\times\mathfrak  C_s, $ then $\mathfrak C_r^\perp\times\mathfrak  C_s^\perp\subseteq \mathfrak C_r\times\mathfrak  C_s.$ This implies $\mathfrak C^\perp_r\subseteq \mathfrak C_r$ and $\mathfrak C^\perp_s\subseteq \mathfrak C_s.$\par
   Similar arguments can be used to prove the converse.
   
\end{proof}

 Now, by using  the above discussion, we present the key theorem of this section.
\begin{theorem}\label{theo-6.8}  Let   $\mathfrak{D} =\mathfrak  C_r\times \mathfrak C_s$ be  a separable   $(\theta, \Theta)$-cyclic code of block length $(r,s)$ over $\mathbb F_qR$. If   $\mathfrak C_r^\perp\subseteq \mathfrak C_r$ and  $\mathfrak C_s^\perp\subseteq \mathfrak C_s$, then there exists a $[[r+4s,2k-(r+4s),d_H]]_q$ QECC  over  $\mathbb F_q$,  where $k$ is dimension  and  $d_H$ is  Hamming distance of  $\Phi(\mathfrak{D})$.
\end{theorem}
\begin{proof}
Suppose  $\mathfrak C_r^\perp\subseteq \mathfrak C_r$ and  $\mathfrak C_s^\perp\subseteq \mathfrak C_s$, then by Theorem \ref {theo-6.6}, we get $\mathfrak D^\perp\subseteq \mathfrak D.$  Now, by Theorem \ref {theo-6.1}, we have $|\Phi(\mathfrak{D}^\perp)| = |\Phi(\mathfrak{D})^\perp|$, then  $\Phi(\mathfrak D)^\perp\subseteq \Phi(\mathfrak D). $ Again by  Theorem \ref{theo-6.2}, we  infer that there exists a  $[[r+4s,2k-(r+4s),d_H]]_q$ QECC   over $\mathbb F_q.$
    
\end{proof} 
 In order to support our findings, we present some illustrative examples by  constructing   QECCs.
\begin{example} \label{ex-6.8}\em
 Let $q=9,~r=26,~ s=6$ and $R= \mathbb F_{9}+u\mathbb F_{9}+v\mathbb F_{9}+uv \mathbb F_{9},$ where $\mathbb F_{9}=\mathbb F_3[\omega]$ with $\omega^2+2\omega+2=0.$ The Frobenius automorphism $\theta:\mathbb F_{9} \to \mathbb F_{9}$ is defined as $\theta(\alpha)= \alpha^3$ for all $\alpha\in \mathbb F_{9}.$ We   have  $\ord(\theta) =2=\ord(\Theta) $  such that  $\ord(\theta)\mid r$,  $\ord(\Theta)\mid s$.   Since $\mathbb F_9[x:\theta]$ is not a UFD,  so both the polynomials $x^{26}-1$ and $x^{6}-1$ have more than one factorization. Let us delve into one of these factors.
 \begin{align*}
 x^{26}-1&= (x^3 + \omega x^2 + 2x\omega + 2)  (x^3 + \omega^7x^2 + 2x + 1)  (x^3 + 2x^2 + \omega^7x + 1)  (x^3 + x^2 + \omega^7x + 2)\times  \\ 
 &\hspace{.5cm}(x^3 + \omega x^2 +\omega^6x + 1) (x^3 + \omega^7x^2 + \omega^6x + 2) (x^3 + x^2 + \omega^3x + \omega^6)\times\\
 &\hspace{.5cm}(x^3 + 2x^2 + \omega^3x + \omega^2)  (x + \omega^2)^2
 \end{align*}
 and 
 \[x^{6}-1=(x+\omega^{6})(x+1)(x+\omega^{2})^2(x+2)(x+\omega^2).\]
 Let     $\ell(x)=(x^3 + 2x^2 + \omega^3x + \omega^2)  (x + \omega^2).$ Then  
 \begin{align*}
 f(x)&= x^{22} + \omega^7x^{21} + \omega^6x^{20} + \omega^3x^{19} + \omega^6x^{18} + \omega^2x^{16} + \omega^3x^{15}  + 2x^{14} 
    + \omega^7x^{13} + \omega^3x^{12}+\\ 
    &\hspace{.5 cm} \omega^7x^{11} 
    + x^{10} + \omega^3x^9 + \omega^2x^8 + \omega^7x^7+ \omega x^6 + \omega^3x^5 + \omega^3x^3 + \omega^7x^2 + \omega^3x + 1\\
 f^\dag(x)&=x^{22} + \omega x^{21} + \omega^7x^{20} + \omega x^{19} + \omega x^{17} + \omega x^{16} + \omega^5x^{15}  + \omega^2x^{14} 
    + \omega x^{13} + x^{12} + \omega^5x^{11}+\\
    & \hspace{.5cm}\omega^3x^{10}+\omega^5x^9 + 2x^8 + \omega x^7 + \omega^2x^6+ \omega^6 x^4 + \omega x^3 + \omega^6x^2 + \omega^5x  + 1\\
 f^\dag(x)f(x)&= ( x^{18} + \omega^6 x^{17} + \omega^5 x^{16} + \omega^2x^{15} 
+ \omega^7x^{14} 
    + \omega^2 x^{12} + 2x^{10}+\omega^2x^9 + x^8 + \omega^6x^6+\\
    &\hspace{.5cm} \omega^3 x^4 + \omega ^2x^3 + \omega x^2 + \omega^6x  + 1)(x^{26}-1).
     \end{align*}
     Next, let us choose  $t_1(x)= (x+1)(x+\omega^6)$, then 
     \begin{align*}
         s_1(x)&= x^4 + \omega^3x^3 + \omega x^2 + \omega^7x + \omega^6
\\
         s^\dag_1(x)&= \omega^6x^4 + \omega^5x^3 + \omega x^2 + \omega x + 1\\
         s^\dag_1(x)s_1(x)&=(\omega^6x^2 + \omega^2
)(x^{6}-1).
     \end{align*}
Finally, we  choose   $t_i(x)=(x+w^{6})$  for $i=2,3,4,$ then 
\begin{align*}
    s_i(x)&= x^5 + \omega^6x^4 + x^3 + \omega^6x^2 + x + \omega^6\\
    s^\dag_i(x)&=  \omega^2x^5 + x^4 + \omega^2x^3 + x^2 + \omega^2x + 1\\
    s^\dag_i(x)s_i(x)&= (\omega^2x^4 + \omega^2x^2 + \omega^2)(x^{6}-1).
\end{align*}
 Here, $\mathfrak C_r =\langle \ell(x)\rangle$ is  a $\theta$-cyclic code  having parameters $[26,22,3]$ over $\mathbb F_{9}$ and $\mathfrak C_s= \langle\kappa_1t_1(x)+ \kappa_2t_2(x)+\kappa_3t_3(x)+\kappa_4t_4(x)\rangle$ is  a $\Theta$-cyclic code   over $R$ such that   $\Phi_1(\mathfrak C_s)$ has parameters $[24,19,3]$ over $\mathbb F_9.$
Therefore, by Proposition \ref{prop-3.2},  we get  $\Phi(\mathfrak{D})$ is a  $[50,41,3]$ linear code over $\mathbb F_{9}$. From the above equation, we have  $x^{26}-1\mid_r~(f^\dag(x)f(x))$ and $x^{6}-1\mid_r~(s^\dag_i(x)s_i(x))$ for $i=1,2,3,4.$ Therefore, by     Theorems \ref{theo-6.4} and \ref{theo-6.5}, we get  that $\mathfrak C_r^\perp\subseteq\mathfrak C_r$ and $\mathfrak C_s^\perp\subseteq \mathfrak C_s$.
Hence,  from Theorem \ref{theo-6.6},  we conclude that  $\mathfrak{D}^\perp\subseteq\mathfrak{D}.$ Thus, Theorems \ref{theo-6.8} enables us to construct a  $[[50,32,3]]_{9}$  QECC. This QECC has better parameters than $[[50,30,3]]_{9}$ as given  in  \cite{E20}. 
\end{example}

\begin{example}\label{ex-6.9}\em

 Let $q=25,~r=8,~ s=10$ and $R= \mathbb F_{25}+u\mathbb F_{25}+v\mathbb F_{25+}uv \mathbb F_{25},$ where $\mathbb F_{25}=\mathbb F_5[\omega]$ with $\omega^2+4\omega+2=0.$ The Frobenius automorphism $\theta:\mathbb F_{25} \to \mathbb F_{25}$ is defined as $\theta(\alpha)= \alpha^5$ for all $\alpha\in \mathbb F_{25}.$ We   have  $\ord(\theta)=2=\ord(\Theta) $ such that  $\ord(\theta)\mid r$,  $\ord(\Theta)\mid s$.   Since $\mathbb F_{25}[x:\theta]$ is not a UFD,  so both the polynomials $x^{8}-1$ and $x^{10}-1$ have more than one factorization. Let us delve into one of these factors.

 \[ x^8-1= (x+1)(x+4)(x+\omega^9)^2(x+\omega^3)^2(x+3)(x+2)\]
 and 
 \[x^{10}-1=(x+4)(x+\omega^{20})(x+\omega^{16})(x+4)(x+1)^2 (x+\omega^8)(x+\omega^{16})(x+\omega^{20})(x+\omega^4).\]
 Let    $\ell(x)=(x+\omega^9)(x+2)$.  Then  
 \begin{align*}
 f(x)&= x^6+\omega^4x^5+\omega^{10}x^4+\omega^4x^3+\omega^{22}x+\omega^{21}\\
 f^\dag(x)&=\omega^{21}x^6+\omega^{14}x^5+\omega^5x^4+\omega^{20}x^3+\omega^{10}x^2+\omega^{20}x+1\\
 f^\dag(x)f(x)&= (\omega^{21}x^4+2x^3+3x+w^9)(x^8-1).
     \end{align*}
     Next, let us choose $t_1(x)= (x+\omega^8)(x+\omega^{16})(x+\omega^{20}),$ then 
     \begin{align*}
         s_1(x)&= x^7+\omega^4x^6+2x^5+\omega^{10}x^4+3x^3+\omega^{22}x^2+4x+\omega^{16}\\
         s^\dag_1(x)&= \omega^8x^7+4x^6+\omega^{14}x^5+3x^4+\omega^2x^3+2x^2+\omega^{20}x+1\\
         s^\dag_1(x)s_1(x)&=(\omega^8x^4+\omega^8x^3+\omega^{15}x^2+\omega^{20}x+\omega^4)(x^{10}-1).
     \end{align*}
Finally, we  choose   $t_i(x)=(x+w^{20})$  for $i=2,3,4,$ then 
\begin{align*}
    s_i(x)&= x^9+\omega^{16}x^8+x^7+\omega^{16}x^6+x^5+\omega^{16}x^4+x^3+\omega^{16}x^2+x+\omega^{16}\\
    s^\dag_i(x)&=  \omega^8x^9+x^8+\omega^{8}x^7+x^6+\omega^{8}x^5+x^4+\omega^{8}x^3+x^2+\omega^{8}x+1\\
    s^\dag_i(x)s_i(x)&= (\omega^8x^8+\omega^{20}x^7+\omega^7x^6+\omega^2x^5+\omega^{21}x^4+\omega^{14}x^3+\omega^{23}x^2+\omega^8x+\omega^4)(x^{10}-1).
\end{align*}
 Here,  $\mathfrak C_r =\langle \ell(x)\rangle$ is a   $\theta$-cyclic code  having parameters $[8,6,3]$ over $\mathbb F_{25}$, $\mathfrak C_s= \langle\kappa_1t_1(x)+ \kappa_2t_2(x)+\kappa_3t_3(x)+\kappa_4t_4(x)\rangle$ is a $\Theta$-cyclic code over $R$  such that   $\Phi_1(\mathfrak C_s)$ has  parameters $[40,34,3]$ over $\mathbb F_{25.}$ Therefore by Proposition \ref{prop-3.2}, we get that  $\Phi(\mathfrak{D})$ is a $[48,40,3]$  linear code  over $\mathbb F_{25}$.  From this, we have $ x^{8}-1 \mid_r~(f^\dag(x)f(x))$ and $ x^{10}-1\mid _r~(s^\dag_i(x)s_i(x))$ for $i=1,2,3,4.$ Therefore, by  Theorems \ref{theo-6.4} and \ref{theo-6.5}, we get  that $\mathfrak C_r^\perp\subseteq\mathfrak C_r$ and $\mathfrak C_s^\perp\subseteq \mathfrak C_s$. Hence, from  Theorem \ref{theo-6.6}, we get that  $\mathfrak{D}^\perp\subseteq\mathfrak{D}.$ Thus, Theorem \ref{theo-6.8} enables us to construct a $[[48,32,3]]_{25}$ QECC. This QECC has better parameters than $[[48,16,3]]_{25}$ as given in \cite{DBAPBC23}. 
\end{example}

As an application of our study, in Table  \ref{tab-2} we  showcase several  MDS QECCs obtained from $\theta$-cyclic codes over $\mathbb F_q$.  The coefficients of generator polynomials are presented in ascending order. For instance, $ (\omega^7,\omega^5,2,1)$ corresponds to the polynomial $w^7+w^5x+2x^2+x^3.$  In Table \ref{tab-3}, we obtain some new QECCs from $\Theta$-cyclic code over $R,$ which have better parameters than the existing quantum codes. Furthermore, in Table \ref{tab-4}, we construct some new QECCs from $(\theta, \Theta )$-cyclic codes  over $\mathbb F_qR.$

\begin{table}[ht]
    \centering
    \caption{MDS QECCs from $\theta$-cyclic code over $\mathbb F_q$}
    \label{tab-2}
    \vspace{.5cm}
    \begin{tabular}{|c|c|c|c|}
    \hline
    $q$\hspace{1.5cm} &$n$\hspace{1.5cm} &$\ell(x)$\hspace{1.5cm} &$[[n,k,d]]_q$\\
    \hline
         $7^2$\hspace{1.5cm} &$14$\hspace{1.5cm} &$(\omega^{36},\omega^{44},\omega^{19},1)$\hspace{1.5cm} &$[[14,8,4]]_{7^2}$\\
   \hline
         $5^3$\hspace{1.5cm} &$15$\hspace{1.5cm} &($\omega^{2},3,\omega^{76},1)$\hspace{1.5cm} &$[[15,9,4]]_{5^3}$\\
         \hline
         $5^2$\hspace{1.5cm} &$16$\hspace{1.5cm} &$(\omega^{5},\omega^{7},\omega^{20},1)$\hspace{1.5cm} &$[[16,10,4]]_{5^2}$\\
         
         \hline
        
         $7^3$\hspace{1.5cm} &$21$\hspace{1.5cm} &($\omega^{261},\omega^{17},\omega^{178},1)$\hspace{1.5cm} &$[[21,15,4]]_{7^3}$\\
         \hline
         $5^5$\hspace{1.5cm} &$25$\hspace{1.5cm} &$(\omega^{1374},\omega^{3116},\omega^{3015},1)$\hspace{1.5cm} &$[[25,19,4]]_{5^5}$\\
         \hline
           $13^2$\hspace{1.5cm} &$26$\hspace{1.5cm} &$(1,\omega^{21},\omega^{148},1)$\hspace{1.5cm} &$[[26,20,4]]_{13^2}$\\
         \hline
    \end{tabular}

\end{table}
\begin{landscape}
\begin{table}[ht]
\centering
\caption{ New QECCs constructed from $\Theta$-cyclic code $\mathfrak C_s$ over  $ R$}
\label{tab-3}
\vspace{.5 cm}

\begin{tabular}
{|c|c|c|c|c|c|c|c|c|c|}
\hline

$q$&$n$  &$t_1(x)$ & $t_2(x)$ & $t_3(x)$ & $t_4(x)$ & $\Phi_1(\mathfrak C_s)$& New QECCS &$\mbox{Existing QECCs}$\\
\hline

9 & 12  & $(\omega^{3},\omega,1,2,1)$ & $(\omega^6,1)$ & $(1,1)$ & $(1)$ & $[48,42,4]$ & $[[48,36,4]]_9$ & $[[48,34,4]]_9$\cite{RGS23} \\
\hline
 9 & 14  & $(\omega^{36},\omega^{44},\omega^{19},1)$ & $(\omega^{18},1)$ & $(\omega^6,1)$ & $(1)$ & $[56,51,4]$ & $[[56,46,4]]_{49}$ & $[[56,44,4]]_{49}$\cite{VPSI23} \\
\hline
25 & 16  & $(\omega^{5},\omega^{7},\omega^{20},1)$ & $(\omega^{22},1)$ & $(\omega^{22},1)$ & $(1)$ & $[64,59,4]$ & $[[64,54,4]]_{25}$ & $[[64,52,4]]_{25}$\cite{VPSI23} \\
\hline
121 & 22  & $(\omega^{100},\omega^{80},\omega^{38},\omega^{20},1)$ & $(\omega^{40},1)$ & $(\omega^{80},1)$ & $(1,1)$ & $[88,81,4]$ & $[[88,74,4]]_{121}$ & $[[88,64,4]]_{121}$\cite{BDUBY20} \\
\hline
121 & 22  & $(\omega^{20},\omega^{107},\omega^{32},1)$ & $(1,1)$ & $(\omega^{80},1)$ & $(\omega^{40},1)$ & $[88,82,3]$ & $[[88,76,3]]_{121}$ & $[[88,72,3]]_{121}$\cite{BDUBY20} \\
\hline
169 & 26  & $(1,\omega^{148},\omega^{21},1)$ & $(\omega^{60},1)$ & $(\omega^{48},1)$ & $(1)$ & $[104,99,4]$ & $[[104,94,4]]_{169}$ & $[[104,80,4]]_{169}$\cite{BDUBY20} \\
\hline
25 & 40  & $(\omega^{80},\omega^{39},\omega^{3},\omega^{83},1)$ & $(\omega^{104},2,\omega^{41},\omega^{100},1)$ & $(4,1)$ & $(1,1)$ & $[160,150,4]$ & $[[160,140,4]]_{25}$ & $[[160,134,4]]_{25}$\cite{RGS23} \\
\hline

\end{tabular}
\end {table}
\end{landscape}


\begin{landscape}
\begin{table}[ht]
\centering
\caption{ New QECCs constructed from  $(\theta,\Theta)$-cyclic code $\mathfrak D$ over  $ \mathbb F_qR$}
\label{tab-4}
\vspace{.5 cm}

\begin{tabular}
{|c|c|c|c|c|c|c|c|c|c|}
\hline
$q$ & $(r,s)$ & $\ell(x)$ & $t_1(x)$ & $t_2(x)$ & $t_3(x)$ & $t_4(x)$ &  New QECCS & \mbox{Existing QECCs}\\
\hline

$27$ & $(18,9)$ & $(\omega^{18}, \omega^{19}, 1)$ & $(\omega^{18}, \omega^{19}, 1)$ & $(\omega^{9}, 1, 1)$ & $(\omega^3, 1, 1)$ & $(1)$ &  $[[54,42,3]]_{27}$ & $[[54,36,3]]_{27}\cite{RGS23}$\\
\hline

$49$ & $(8,16)$ & $(\omega^{12},\omega^{11},\omega^{23}, 1)$ & $(\omega^{12},\omega^{11},2,1)$ & $(\omega^{12},1)$ & $(\omega^{12},1)$ & $(1)$ & $[[72,56,4]]_{49}$ & new\\

\hline
$121$ & $(10,10)$ & $(\omega^{78},1, \omega^{54}, 1)$ & $(\omega^{78},1, \omega^{54}, 1)$ & $(5, 1)$ & $(6, 1)$ & $(1)$ &  $[[50,34,4]]_{121}$ & $[[52,28,4]]_{121}\cite{ALY21}$\\
\hline
$121$ & $(10,20)$ & $(\omega^{78},1, \omega^{54}, 1)$ & $( \omega^{94},\omega^{16},\omega^{118},\omega^{61},1)$ & $(\omega^{49}, 1)$ & $(\omega^{49}, 1)$ & $(\omega^{49},1)$ &  $[[90,70,4]]_{121}$ & new\\
\hline
$49$ & $(48,14)$ & $(\omega^{36},0, 0,\omega^{17},\omega^{28}, 1)$ & $(\omega^{36},\omega^{44},\omega^{19},1)$ & $(\omega^{18},1)$ & $(\omega^6,1)$ & $(1)$ & $[[104,84,4]]_{49}$ & new\\
\hline

$169$ & $(8,24)$ & $(\omega^{117}, \omega^{109}, 1)$ & $(\omega^{15}, \omega^{3}, 1)$ & $(11, 1, 1)$ & $(\omega^{85}, 1, 1)$ & $(1)$ &  $[[104,92,3]]_{169}$ & $[[104,88,3]]_{169}\cite{BDUBY20}$\\
\hline
$25$ & $(16,24)$ & $(\omega^{5},\omega^{7}, \omega^{20}, 1)$ & $( \omega^{7}, 1)$ & $(\omega^2, 2, 1)$ & $(\omega^{9}, 1, 1)$ & $(\omega^{13},1)$ &  $[[112,94,4]]_{25}$ & new\\
\hline
$25$ & $(24,24)$ & $(\omega^{36},\omega^{13}, \omega, 1)$ & $( \omega^{7}, 1)$ & $(\omega^2, 2, 1)$ & $(\omega^{9}, 1, 1)$ & $(\omega^{13},1)$ &  $[[120,102,4]]_{25}$ & $[[120,96,4]]_{25}\cite{RGS23}$\\
\hline
$49$ & $(42,24)$ & $(\omega^{2},1, 0,0,\omega^{18}, 1)$ & $( \omega, 1)$ & $(\omega^15, \omega^23, 1)$ & $(\omega^{22}, \omega^{46}, 1)$ & $(\omega^{46},1)$ & $[[138,116,4]]_{49}$ & new\\
\hline
$25$ & $(20,40)$ & $(\omega^{9},\omega^{80}, \omega^{51},\omega, 1)$ & $(\omega^{80},\omega^{39},\omega^{3},\omega^{83},1)$ & $(\omega^{104},2,\omega^{41},\omega^{100},1)$ & $(1, 1)$ & $(4,1)$ &  $[[180,160,4]]_{25}$ & $[[180,150,4]]_{25}\cite{RGS23}$\\
\hline

\end{tabular}
\end {table}
\end{landscape}

\section{Conclusion}

In summary, this paper investigates the algebraic structure of $(\theta,\Theta)$-cyclic codes with block length $(r,s)$ over the ring $\mathbb F_qR$, where $R= \mathbb F_q+u\mathbb F_q+v\mathbb F_q+uv\mathbb F_q$ with $u^2=u, v^2=v, uv=vu$, and $q$ is an odd prime power. Our exploration begins with the idempotent decomposition of the ring $R$, followed by an in-depth examination of linear codes over the ring $\mathbb F_qR$. We define a Gray map from $\mathbb F_q^r\times R^s\to \mathbb F_q^{r+4s}$ and investigate its fundamental properties. The algebraic structure of $(\theta,\Theta)$-cyclic codes with block length $(r,s)$ over $\mathbb F_qR$ is thoroughly determined, and we establish generator polynomials for this family of codes in Theorem \ref{theo-4.4}. The algebraic structure of separable codes is discussed in detail in Theorem \ref{theo-4.7}, and the minimal generating set of $(\theta,\Theta)$-cyclic codes over $\mathbb F_qR$ is presented in Theorem \ref{theo-4.12}. The relationship between generator polynomials of $(\theta,\Theta)$-cyclic codes and their duals is established in Theorem \ref{theo-5.11}. As an application of our findings, we demonstrate the construction of QECCs from separable $(\theta,\Theta)$-cyclic codes over $\mathbb F_qR$ in Theorem \ref{theo-6.8}. Additionally, detailed examples (Examples \ref{ex-6.8}, \ref{ex-6.9} ) are provided to illustrate the process of constructing QECCs. Our study yields optimal and near-optimal codes from $\Theta$-cyclic codes over $R$ (Table \ref{taboptimal-1}), MDS QECCs from $\theta$-cyclic codes over $\mathbb F_q$ (Table \ref{tab-2}), and new QECCs with better parameters compared to existing codes (Tables \ref{tab-3} and \ref{tab-4}).
\section*{Acknowledgement}
The authors are grateful for the considerate and helpful comments given by the anonymous reviewers. These comments have played a crucial role in improving the quality of the manuscript.
The first author is thankful to the University Grant Commission (UGC), Govt. of India.  The third author is thankful to the DST-INSPIRE, and the fourth author is thankful to NBHM, Department of Atomic Energy, for their financial support. The work of  AKU is partially supported by a research grant under the MATRICS scheme MTR/2020/000006. Sachin Pathak is thankful to NIIT University for the NUIG research grant.
\section*{Declaration}
\subsection*{Conflict of interest}
The authors declare no conflict of interest.
\subsubsection*{Data availability}
The authors declare that all data supporting the findings of this study are available within the article.

\end{document}